%% file: main.tex
\documentclass[a4paper,UKenglish,cleveref, autoref]{lipics-v2019}
\hideLIPIcs
\title{Path Contraction Faster than $2^n$}%\footnote{Due to space limitations most proofs have been moved to the Appendix.}}
\titlerunning{Path Contraction Faster than $2^n$}
\author{Akanksha Agrawal}{Indian Institute of Techonology Madras,India}{akanksha@cse.iitm.ac.in}{}{}
\author{Fedor V. Fomin}{University of Bergen, Bergen, Norway}{fomin@ii.uib.no}{}{}
\author{Daniel Lokshtanov}{University of California Santa Barbara, Santa Barbara, California}{daniello@ ucsb.edu}{}{}
\author{Saket Saurabh}{Institute of
Mathematical Sciences, HBNI, Chennai, India \and University of Bergen, Bergen, Norway \and UMI ReLax}{saket@imsc.res.in}{}{This work is supported by the European Research Council (ERC) via grant LOPPRE, reference no. 819416.}
\author{Prafullkumar Tale}{Indian Institute of Science
Education and Research, Pune, India}{prafullkumar@iiserpune.ac.in}{}{}

\authorrunning{A. Agrawal, F. Fomin, D. Lokshtanov, S. Saurabh, and P. Tale}%TODO mandatory. First: Use abbreviated first/middle names. Second (only in severe cases): Use first author plus 'et al.'

\Copyright{A. Agrawal, F. Fomin, D. Lokshtanov, S. Saurabh and P. Tale}%TODO mandatory, please use full first names. LIPIcs license is "CC-BY";  http://creativecommons.org/licenses/by/3.0/

%TODO mandatory: Please choose ACM 2012 classifications from https://dl.acm.org/ccs/ccs_flat.cfm 

%TODO mandatory; please add comma-separated list of keywordss

\category{}%optional, e.g. invited paper

\relatedversion{An extended abstract of this article
appeared in \href{https://drops.dagstuhl.de/entities/document/10.4230/LIPIcs.ICALP.2019.11}{ICALP 2019} and full version appeared in \href{https://epubs.siam.org/doi/10.1137/19M1259638}{SIDMA 2020}.}%optional, e.g. full version hosted on arXiv, HAL, or other respository/website
\supplement{}%optional, e.g. related research data, source code, ... hosted on a repository like zenodo, figshare, GitHub, ...

\acknowledgements{During the part of the project, the first author was a Post-doctoral candiate at Hungarian Academy of Sciences, Budapest, Hungary and the last author was a PhD student at  Institute of
Mathematical Sciences, HBNI, Chennai, India.}%optional

\nolinenumbers %uncomment to disable line numbering

\hideLIPIcs  %uncomment to remove references to LIPIcs series (logo, DOI, ...), e.g. when preparing a pre-final version to be uploaded to arXiv or another public repository

\EventEditors{Christel Baier, Ioannis Chatzigiannakis, Paola Flocchini, and Stefano Leonardi}
\EventNoEds{4}
\EventLongTitle{46th International Colloquium on Automata, Languages, and Programming (ICALP 2019)}
\EventShortTitle{ICALP 2019}
\EventAcronym{ICALP}
\EventYear{2019}
\EventDate{July 9--12, 2019}
\EventLocation{Patras, Greece}
\EventLogo{eatcs}
\SeriesVolume{132}
\ArticleNo{6}

\newcommand{\defproblemout}[3]{
  \vspace{1mm}
\noindent\fbox{
  \begin{minipage}{0.96\textwidth}
  \begin{tabular*}{\textwidth}{@{\extracolsep{\fill}}lr} #1 \\ \end{tabular*}
  {\bf{Input:}} #2  \\
  {\bf{Output:}} #3
  \end{minipage}
  }
  \vspace{1mm}
}

\newcommand{\defproblem}[3]{
  \vspace{1mm}
\noindent\fbox{
  \begin{minipage}{0.96\textwidth}
  \begin{tabular*}{\textwidth}{@{\extracolsep{\fill}}lr} #1 \\ \end{tabular*}
  {\bf{Input:}} #2  \\
  {\bf{Question:}} #3
  \end{minipage}
  }
  \vspace{1mm}
}

\usepackage{complexity} % For symbols in complexity class.
\usepackage[ruled,vlined,algosection,linesnumbered]{algorithm2e}%algopart %algosection %algochapter
\usepackage{todonotes}
\usepackage{enumitem}
\usepackage{algorithmic}
\newtheorem{observation}{Observation}[section]

\newtheorem{lemma1}{Lemma}[section]

\newcommand{\es}{\texttt{ES}}
\newcommand{\os}{\texttt{OS}}
\newcommand{\yes}{yes}

\newcommand{\calA}{\mathcal{A}}

\newcommand{\calO}{\ensuremath{{\mathcal O}}}
\newcommand{\OO}{\mathcal{O}}

\newcommand{\calS}{\mathcal{S}}

\newcommand{\calW}{\mathcal{W}}

\let\oldnl\nl% Store \nl in \oldnl
\newcommand{\nonl}{\renewcommand{\nl}{\let\nl\oldnl}}% Remove line number for one line
\newcommand{\mcal}[1]{\mathcal{#1}} 

\newcommand{\TDCS}{{\sc $3$-DCS}}
\newcommand{\STDCS}{{\sc Small $3$-DCS}}
\newcommand{\PFC}{\textsc{$P_5$-Contraction}}
\newcommand{\what}{\widehat}
\newcommand{\SOEPC}{{\sc Small Odd/Even PC}} 
\newcommand{\EPPC}{{\sc Enum-Partial-PC}} 
\newcommand{\BPC}{{\sc Balanced PC}} 
\newcommand{\TDCPC}{{\sc $2$-Union Heavy PC}} 
\newcommand{\ThDCPC}{{\sc Near Small Odd/Even PC}}
\newcommand{\NSOEPC}{{\sc Near Small Odd/Even PC}}
\newcommand{\runtime}{1.99987}
\newcommand{\Hard}{{hard}}

\begin{document}

\maketitle

\begin{abstract}
A graph $G$ is contractible to a graph $H$ if there is a set $X \subseteq E(G)$, such that $G/X$ is isomorphic to $H$. Here, $G/X$ is the graph obtained from $G$ by contracting all the edges in $X$. For a family of graphs $\cal F$, the $\mathcal{F}$-\textsc{Contraction} problem takes as input a graph $G$ on $n$ vertices, and the objective is to output the largest integer $t$, such that $G$ is contractible to a graph $H \in {\cal F}$, where $|V(H)|=t$. When $\cal F$ is the family of paths, then the corresponding $\mathcal{F}$-\textsc{Contraction} problem is called \textsc{Path Contraction}. The problem \textsc{Path Contraction} admits a simple algorithm running in time $2^{n}\cdot n^{\OO(1)}$. In spite of the deceptive simplicity of the problem, beating the $2^{n}\cdot n^{\OO(1)}$ bound for \textsc{Path Contraction} seems quite challenging. In this paper, we design an exact exponential time algorithm for \textsc{Path Contraction} that runs in time ${\runtime^n}\cdot n^{\calO(1)}$. We also define a problem called \textsc{$3$-Disjoint Connected Subgraphs}, and design an algorithm for it that runs in time $1.88^n\cdot n^{\calO(1)}$. The above algorithm is used as a sub-routine in our algorithm for {\sc Path Contraction}.  
\end{abstract}

\input{intro_path.tex}

\input{prelims_path.tex}
\input{enum-conn-sets.tex}
\input{prob-disjoint-conn-3.tex}

\input{exact_algo_path.tex}
\input{conclusion_path.tex}

%%
%% Bibliography
%%

%% Please use bibtex, 

\bibliography{references}

\end{document}

%% file: intro_path.tex
\section{Introduction}
\label{sec:intro}
Graph editing problems are one of the central problems in graph theory that have received a lot of attention in algorithm design. Some of the natural graph editing operations are vertex/edge deletion, edge addition, and edge contraction. For a family of graphs $\cal F$, the {\sc $\cal F$-Editing} problem takes as input a graph $G$, and the objective is to find the minimum number of operations required to transform $G$ into a graph from $\cal F$. In fact, the {\sc $\cal F$-Editing} problem, where the edit operations are restricted to one of vertex deletion, edge deletion, edge addition, or edge contraction have also received
a lot of attention in algorithm design. The {\sc $\cal F$-Editing} problems encompass several classical \NP-\Hard\ problems like \textsc{Vertex Cover}, \textsc{Feedback Vertex Set}, \textsc{Longest Path}, etc. 

The {\sc $\cal F$-Editing} problem where the only allowed edit operation is edge contraction, is called {\sc$\cal  F$-Contraction}. For a graph $G$ and an edge $e=uv\in E(G)$, {\em contraction} of an edge $uv$ in $G$ results in a graph $G/e$, which is obtained by deleting $u$ and $v$ from $G$, adding a new vertex $w_e$ and making $w_e$ adjacent to the neighbors of $u$ or $v$ (other than $u,v$). A graph $G$ is \emph{contractible} to a graph $H$, if there exists a subset $X\subseteq E(G)$, such that if we contract each edge from $X$, then the resulting graph is isomorphic to $H$. For several families of
graphs $\cal F$, early papers by Watanabe et al.~\cite{watanabe81,watanabe1983np} and Asano and Hirata~\cite{asano1983edge} showed that {\sc $\cal F$-Contraction} is \NP-\Hard. The \NP-\Hard ness of problems like {\sc Tree Contraction} and {\sc Path Contraction}, which are the {\sc $\cal F$-Contraction} problems for the family of trees and paths, respectively, follows easily from~\cite{asano1983edge,brouwer1987contractibility}. A restricted version of \textsc{Path Contraction}, is the problem \textsc{$P_t$-Contraction}, where $t$ is a fixed constant. \textsc{$P_t$-Contraction} is shown to be \NP-\Hard\ even for $t=4$, while for $t\leq 3$, the problem is polynomial time solvable~\cite{brouwer1987contractibility}. \textsc{$P_t$-Contraction} alone had received lot of attention for smaller values of $t$, even when the input graph is from a very structured family of graphs (for instance, see ~\cite{brouwer1987contractibility,van2009partitioning,HeggernesHLP14,dabrowski2017contracting,fiala2013note,kern2018contracting}, and the references therein). 

Several \NP-\Hard\ problem like {\sc SAT}, {\sc $k$-SAT}, {\sc Vertex Cover}, {\sc Hamiltonian Path}, etc. are known to admit an algorithm running in time $\calO^\star(2^n)$\footnote{The $\calO^{\star}$ notation hides polynomial factors in the running time expression.}. These results are obtained by techniques like brute force search, dynamic programming over subsets, etc. One of the main questions that arise in this context is: can we break the $\calO^\star(2^n)$ barrier for these problems. In fact, the hardness of {\sc SAT} gives rise to the Strong Exponential Time Hypothesis (SETH) of Impagliazzo and Paturi~\cite{impagliazzo2001problems,DBLP:journals/jcss/ImpagliazzoP01}, which rules out existence of $\calO^\star((2-\epsilon)^n)$-time algorithm for {\sc SAT}, for any $\epsilon>0$. SETH has been used to obtain such algorithmic lower bounds for many other \NP-\Hard\ problems (see for example,~\cite{DBLP:journals/talg/CyganDLMNOPSW16,DBLP:journals/talg/LokshtanovMS18}). Not all \NP-\Hard\ problems seem to be as ``hard'' as {\sc SAT}. For many \NP-\Hard\ problems, it is possible to break the $\calO^\star(2^n)$ barrier. For instance problems like {\sc Vertex Cover} and (undirected) {\sc Hamiltonian Path} are known to admit algorithms running in time $\calO^\star((2-\epsilon)^n)$, for some $\epsilon>0$~\cite{DBLP:journals/siamcomp/Bjorklund14,DBLP:journals/siamcomp/TarjanT77}. Thus, one of the natural question is for which \NP-\Hard\ problems we can avoid the ``brute force search'', and obtain algorithms that are better than $\calO^\star(2^n)$. 

In this article, we focus on the problem {\sc Path Contraction}, which is formally defined below. 
 
 \defproblemout{\textsc{Path Contraction}}{Graph $G$.}{Largest integer $t$, such that $G$ is contractible to $P_{t}$.}
\vspace{0.2cm}

Note that, if $t$ is the largest integer such that $G$ is contractible to $P_t$ then the minimum number of edge contraction operations is $(n - t)$.
%The \NP-\Hard ness result of {\sc $P_4$-Contraction}~\cite{brouwer1987contractibility} implies that there is no algorithm for the  $\mathcal{O}(n^{f(\ell)})$.
{\sc Path Contraction} is known to admit a simple algorithm that runs in time $\calO^\star(2^n)$. Such an algorithm can be obtained by coloring the input graph with two colors and contracting connected components in the colored subgraphs. For a deceptively simple problem like {\sc Path Contraction}, it seems quite challenging to break the $\calO^\star(2^n)$ barrier. The problem {\sc $2$-Disjoint Connected Subgraphs ($2$-DCS)}\footnote{See Section~\ref{sec:prob-dis-conn-3} for formal definitions.}, can be ``roughly'' interpreted as solving {\sc $P_4$-Contraction}. (We can use the algorithm for {\sc $2$-DCS} to solve {\sc $P_4$-Contraction} within the same time bound.) There have been studies, which break the $\calO^\star(2^n)$ brute force barrier, for {\sc $2$-DCS}. In particular, Cygan et al.~\cite{cygan2014solving} designed a $\calO^{\star}(1.933^n)$ algorithm for {\sc $2$-DCS}. This result was improved by Telle and Villanger, who designed an algorithm running in time $\calO^{\star}(1.7804^n)$, for the problem~\cite{telle2013connecting}. The main goal of this article is to break the $\calO^\star(2^n)$ barrier for {\sc Path Contraction}. Obtaining such an algorithm for {\sc Path Contraction} was stated as an open problem in~\cite{van2009partitioning}. 

%Any connected graph can be contracted to an edge which is a path on two vertices. In this work, we address a question of determining the largest integer $\ell$ for given graph such that it can be contracted to $P_{\ell}$, path on $\ell$ vertices. Formally, we study following problem.

%\noindent \textbf{Previous Results:} Early paper of Brouwer and Veldman states that we can determine  whether a given graph can be contracted to $P_3$ or not in polynomial time but it is \NP-\Hard\ to determine whether it can be contracted to $P_4$ or not \cite{brouwer1987contractibility}. This implies that we can not expect an algorithm for the problem which runs in time $\mathcal{O}(n^{f(\ell)})$. There is a simple algorithm running in time $\calO^{\star}(2^n)$ using $2$-coloring of input graph and checking whether each colored component corresponds to a witness set or not. Algorithm with better running time are known for special case. Cygan et al. \cite{cygan2014solving} observed that \textsc{$P_4$-Contraction} is same as partitioning given graph into two disjoint connected subgraphs which contain specified terminals. They called it $2$-\textsc{Disjoint Connected Subgraphs} problem and gave an algorithm running in time $\calO^{\star}(1.933^n)$. Telle and Villanger \cite{telle2013connecting} presented an algorithm to solve the same problem in time $\calO^{\star}(1.7804^n)$.

\subparagraph{Our Results} We design an algorithm for \textsc{Path Contraction} running in time $\calO^{\star}(\runtime^n)$, where $n$ is the number of vertices in the input graph. To the best of our knowledge, this is the first non-trivial algorithm for the problem, which breaks the $\calO^\star(2^n)$ barrier. To obtain our main algorithm for \textsc{Path Contraction}, we design four different algorithms for the problem, which are used as subroutines to the main algorithm. We exploit the property that certain types of algorithms are better for certain instance, but may be inefficient for certain other instances. Roughly speaking, we look for solutions using different algorithms, and then the best suited algorithm for the instance is used to return the solution. When one of the four algorithms is called as a subroutine, it does not necessarily return an optimum solution for the instance, rather it only looks for solutions that satisfy certain conditions. These conditions are quantified by fractions associated with the input graph. We note that for appropriate values of these ``fractions'', each of our subroutine still serve as an algorithm for {\sc Path Contraction} (and thus, can compute the optimal solution). We argue that there is always a solution which satisfies the conditions for one of the subroutines, by setting the values of the fractions appropriately. A saving over $\calO^\star(2^n)$, in the running time achieved by our algorithm, also exploits the property that ``small'' connected sets with bounded neighborhood can be enumerated ``efficiently''.  %In the following, we roughly sketch the type of solutions each of our subroutine looks for. 

In the following we very briefly explain the type of solutions we look for, in our subroutines. Consider a path $P_t$, such that $G$ can be contracted to $P_t$, where $t$ is the largest such integer. The solution $t$, can be ``witnessed'' by a partition $\calW=\{W_1,W_2,\cdots,W_t\}$ of $V(G)$, where the vertices from $W_i$ ``merge'' to the $i$th vertex of $P_t$ (a formal definition for it can be found in Section~\ref{sec:prelimsxyz}). Such a ``witness'' is called a $P_t$-\emph{witness structure}. The first (subroutine) algorithm for {\sc Path Contraction} searches for a solution where the $P_t$-witness structure can be ``split'' into two connected disjoint parts which are ``small''. Then, it exploits the ``smallness'' of the parts to compute solutions efficiently, and combines them to compute the solution for whole graph. The second subroutine searches for a pair of sets in the $P_t$-witness structure which are very dense. Then it exploit the sparseness of the remaining graph to efficiently compute partial solutions for them. Moreover, the pair of dense parts are resolved using the algorithm of Telle and Villanger for {\sc $2$-Disjoint Connected Subgraph}~\cite{telle2013connecting}. The third routine works with a hope that the total number of vertices in one of odd/even sets from $\calW$ can be bounded. Finally, the fourth subroutine works by exploiting a similar odd/even property as the third subroutine, but it relaxes the condition to ``nearly'' small odd/even set.  

To design our algorithm, we also define a problem called $3$-\textsc{Disjoint Connected Subgraphs ($3$-DCS)}, which is an extension of the $2$-\textsc{Disjoint Connected Subgraphs ($2$-DCS)} problem. {\sc $3$-DCS} takes as input a graph $G$ and disjoint sets $Z_1,Z_2\subseteq V(G)$, and the goal is to partition $V(G)$ into three sets $(V_1,U,V_2)$, such that graphs induced on each of the parts is connected and $Z_i\subseteq V_i$, for $i\in [2]$. We design an algorithm for {\sc $3$-DCS} running in time $\calO^{\star}(1.88^n)$. The fourth subroutine of our algorithm uses the algorithm for {\sc $3$-DCS} as a subroutine. 

%The already known {\sc $2$-DCS} problem is same as {\sc $3$-DCS}, apart from the fact that we want to partition $V(G)$ into two sets rather than three (we do not have the set $U$ in a solution). As was mentioned earlier, {\sc $2$-DCS} is known to admit an algorithm running in time $\calO^{\star}(1.7804^n)$~\cite{telle2013connecting}. Our algorithm uses the algorithm for both {\sc $2$-DCS} and {\sc $3$-DCS} as a subroutine. 

As a corollary to our $\calO^{\star}(1.88^n)$-time algorithm for {\sc $3$-DCS}, we obtain that {\sc $P_5$-Contraction} admits an algorithm running in time $\calO^{\star}(1.88^n)$.

%% file: prelims_path.tex
\section{Preliminaries}\label{sec:prelimsxyz}
In this section, we state some basic definitions and introduce terminologies from graph theory. We use standard terminology from the book of Diestel~\cite{diestel-book} for the graph related
terminologies which are not explicitly defined here. We also establish some notations that will be used throughout. 

We denote the set of natural numbers by $\mathbb{N}$ (including $0$). For $k \in \mathbb{N}$, $[k]$ denotes the set $\{1,2,\ldots, k\}$. %Let $X,Y$ be two sets. For a function $\varphi: X \rightarrow Y$ and $y \in Y$, by $\varphi^{-1}(y)$ we denote the set $\{x \in X \mid \varphi(x)=y\}$.

%In this paper, we consider simple graphs with finite number of vertices. 
%We use standard notation from graph theory \cite{diestel-book}. 
We note that all graphs considered in this article are connected graphs on at least two vertices (unless stated otherwise). For a graph $G$, the sets $V(G)$ and $E(G)$ denote the sets of vertices and edges in $G$, respectively. Two (distinct) vertices $u, v$ in $V(G)$ are \emph{adjacent} if the edge $uv \in E(G)$. For an edge $uv$, the vertices $u$ and $v$ are the \emph{endpoints} of $uv$. The neighborhood of a vertex $v$, denoted by $N_G(v)$, is the set of vertices adjacent to $v$ and its degree $d_G(v)$, is $|N_G(v)|$. For a set $S\subseteq V(G)$, $N_G(S)$ denotes the neighborhood of $S$, i.e., $N_G(S) = (\bigcup_{s\in S}N_G(s)) \setminus S$. The subscripts in the above notations are omitted when the context is clear. 

For a set of edges $F \subseteq E(G)$, $V(F)$ is the set of vertices that are endpoints of edges in $F$. For $S \subseteq V(G)$, we denote the graph obtained by deleting $S$ from $G$ by $G - S$, i.e., the vertex set and edge set of $G-S$ is $V(G)\setminus S$ and $\{uv\in E(G) \mid u,v \not\in S\}$, respectively. Furthermore, the subgraph of $G$ induced by $S$ is the graph $G[S] = G \setminus (V(G) \setminus S)$. For two subsets $S_1, S_2\subseteq V(G)$, we say $S_1, S_2$ are \emph{adjacent} if there exists an edge in $G$ with one endpoint in $S_1$ and the other endpoint in $S_2$.

A path $P_t=(v_1,v_2,\cdots,v_t)$ on $t$ vertices, where $t\in \mathbb{N}$ is the graph with vertex set $\{v_1,v_2,\cdots,v_t\}$ and edge set $\{v_iv_{i+1} \mid i \in [t-1]\}$. Furthermore, $P_t$ is a path between $v_1$ and $v_t$. A graph $G$ is \emph{connected} if for every distinct vertices $u,v\in V(G)$, there is a path (which is subgraph of $G$) between $u$ and $v$. Consider a graph $G$. A (vertex inclusion-wise) maximal connected subgraph of $G$ is a \emph{connected component} or a \emph{component} of $G$. A set $A \subseteq V(G)$ is a \emph{connected set} in $G$ if $G[A]$ is a connected graph.

%Consider an edge $e=uv \in E(G)$. The {\em contraction} of edge $e$ in $G$ deletes the vertices $u$ and $v$ from $G$, and adds a new vertex, which is made adjacent to vertices (other than $u$ and $v$) that were adjacent to at least one of $u$ or $v$. Any parallel edges added in the process are deleted so that the graph remains simple. The resulting graph is denoted by $G/e$. 

Consider a graph $G$ and an edge $e=uv\in E(G)$. The graph obtained after ``contracting'' the edge $e$ in $G$ is denoted by $G/e$. That is, $V(G/e) = (V(G) \cup \{w_{e}\}) \setminus \{u, v\}$ and 
$E(G/e) = \{xy \mid x,y \in V(G) \setminus \{u, v\}, xy \in E(G)\}  \cup \{w_{e}x |\ x \in (N_G(u) \cup N_G(v)) \setminus \{u,v\}\}$, where $w_e$ is a newly added vertex. In the above, for an edge $ux\in E(G) \setminus \{uv\}$, the edge $w_{e}x \in E(G/e)$ is the \emph{renamed} edge of $ux$. For $F \subseteq E(G)$, $G/ F$ denotes the graph obtained from $G$ by contracting each (renamed) edge in $F$. (We note that the order in which the edges in $F$ are contracted is immaterial.) 

A graph $G$ is \emph{isomorphic} to a graph $H$ if there exists a bijective function $\phi : V(G) \rightarrow V(H)$, such that for $v,u \in V(G)$, $uv \in E(G)$ if and only if $(\phi(v), \phi(u)) \in E(H)$. A graph $G$ is \emph{contractible} to a graph $H$ if there exists $F \subseteq E(G)$, such that $G/F$ is isomorphic to $H$. In other words, $G$ is contractible to $H$ if there is a surjective function $\varphi : V(G) \rightarrow V(H)$, with $W(h) =\{v\in V(G) \mid \varphi(v) = h\}$, for $h\in V(H)$, with the following properties:
\setlist{nolistsep}
\begin{itemize}[noitemsep]
\item for any $h\in V(H)$, the graph $G[W(h)]$ is connected, and
\item for any two vertices $h, h' \in V(H)$, $hh' \in E(H)$ if and only if $W(h)$ and $W(h')$ are adjacent in $G$.
\end{itemize}
Let ${\cal W} =\{W(h) \mid h \in V(H)\}$. The sets in ${\cal W}$ are called \emph{witness sets}, and $\cal W$ is an $H$-\emph{witness structure} of $G$. %A witness set that contains more than one vertex is a \emph{big witness set}, and otherwise it is a \emph{small/singleton} witness set.  

%If a graph $G$ has an $H$-witness structure, then graph $H$ can be obtained from $G$ by series of edge contractions. For a fixed $H$-witness structure, let $F$ be union of spanning trees of all witness sets. By convention, spanning tree of a singleton set is an empty set. To obtain graph $H$ from $G$, it is necessary and sufficient to contract edges in $F$. If such witness structure exists then we say graph $G$ is contractible to $H$.

In this paper, we will restrict ourselves to contraction to paths. This allows us to use an ordered notation for witness sets, rather than just the set notation. This ordering of the sets in witness set is given by the ordering of vertices in the path. That is, for a $P_t=(h_1,h_2,\cdots, h_t)$-witness structure, $\calW = \{W(h_1), W(h_2), \cdots, W(h_t)\}$ of a graph $G$, we use the ordered witness structure notation, $(W(h_1), W(h_2), \cdots, W(h_t))$, or simply, $(W_1, W_2, \cdots, W_t)$.

%In this paper, we slightly abuse the notation of set brackets when writing a $P_t$-witness structure of graph. When we say $\calW = \{W_1, W_2, \dots, W_t\}$ is a $P_t$-witness structure of graph $G$, we treat $\calW$ as \emph{ordered set}. In other words, we assume that one end point in path $P_t$ is designated as first vertex and witness sets $W_1, W_2, \dots$ corresponds to first, second, and so on vertices in $P_t$. 

In the following, we give some useful observations regarding contraction to paths. 

\begin{observation}\label{obs:contract-P2} Any connected graph can be contracted to $P_2$.
\end{observation}

\begin{observation}\label{obs:singleton-end-bags} Let $G$ be a graph contractible to $P_t$. Then, there is a $P_t$-witness structure, $\mathcal{W} = (W_1, \dots, W_t)$, of $G$ such that $W_1$ is a singleton set. Moreover, if $t\geq 3$, then there is a $P_t$-witness structure, $\mathcal{W} = (W_1, \dots, W_t)$, of $G$ such that both $W_1$ and $W_t$ are singleton set.
\end{observation}

We end this section with an observation which will be used to bound the number of subsets of a set $U$ which are of size at most $\mu |U|$ for a fixed positive constant $\mu$ which is strictly less than $1/2$. We start with following inequality for integers $n$ and $k$ such that $k \le n$.
$$\binom{n}{k} \le \Big[\Big(\frac{k}{n}\Big)^{-\frac{k}{n}}\cdot \Big(1 - \frac{k}{n}\Big)^{\frac{k}{n}-1} \Big]^n$$
Using above inequality we get following upper bound on summand for $k < n/2$.   
$$\sum_{i=1}^k\binom{n}{i} \le k \cdot \binom{n}{k} \le k \cdot \Big[\Big(\frac{k}{n}\Big)^{-\frac{k}{n}}\cdot \Big(1 - \frac{k}{n}\Big)^{\frac{k}{n}-1} \Big]^n$$
For a positive constant $\mu < 1/2$, such that $k\leq \mu n$, the above inequalities can be written as: 
$$\sum_{i=1}^{\lfloor\mu n\rfloor} \binom{n}{i} \le \mu n \cdot \Big[\mu^{-\mu} \cdot (1 - \mu)^{\mu-1} \Big]^n $$
\begin{equation} \label{eq:enumeration-over-subsets}
\sum_{i=1}^{\lfloor\mu n\rfloor} \binom{n}{i} \le \mu n \cdot \Big[\frac{1}{\mu^{\mu}} \cdot \frac{1}{(1 - \mu)^{1-\mu}} \Big]^n  = \mu n [{g(\mu)}]^n
\end{equation}
where function $g(\mu)$ is defined as:
$$g(\mu) = \frac{1}{\mu^{\mu} \cdot (1-\mu)^{(1-\mu)}}$$

Following observation is implied by above inequalities.

\begin{observation}\label{obs:subset-no} For a set $U$ with $n$ elements and a constant $\mu < 1/2$, the number of subsets of $U$ of size at most $\mu n$ is bounded by $\calO^{\star}([g(\mu)]^n)$. Moreover, all such subsets can be enumerated in time $\calO^{\star}([g(\mu)]^n)$.
\end{observation}

%% file: enum-conn-sets.tex
%\section{Enumeration of Connected Sets}
%\label{sec:enumeration-of-conn-sets}

%A graph is called \emph{connected} if there is a path between every pair of vertices. A maximal connected subgraph is called a \emph{connected component} or a \emph{component} in a graph. A set $A \subseteq V(G)$ is a \emph{connected set} in $G$ if $G[A]$ is a connected graph. 

For a graph $G$, a non-empty set $Q\subseteq V(G)$, and integers $a,b \in \mathbb{N}$, a connected set $A$ in $G$ is a $(Q,a,b)$-\emph{connected set} if $Q \subseteq A$, $|A| = a$, and $|N(A)|\leq b$. Moreover, a connected set $A$ in $G$ is an $(a,b)$-\emph{connected set} if $|A|\leq a$ and $|N(A)|\leq b$. Next, we state results regarding $(Q,a,b)$-connected sets and connected sets, which follow from Lemma~$3.1$ of~\cite{DBLP:journals/combinatorica/FominV12}. (We note that their result gives slightly better bounds, but for simplicity, we only use the bounds stated in the following lemmas.)

\iffalse
\begin{algorithm}[t]
  \KwIn{A graph $G$, a non-empty set $Q \subseteq V(G)$, and integers $a,b \in \mathbb{N}$.}
  \KwOut{The set of all $(Q,a,b)$-connected sets in $G$.}
  
    \If{$|Q| > a$ or $|N[Q]|> a+b$}
    {
	\nonl    \Return $\emptyset$
	}
	\If{$|Q| = a$ and $G[Q]$ is connected}
    {
	 \nonl   \Return $\{Q\}$
	}
	\If{$|Q| = a$ and $G[Q]$ is not connected} 
    {
	   \nonl \Return $\emptyset$
	}	
	
		\If{$N(Q)=\emptyset$ and $|Q| < a$}
    {

    \nonl  \Return $\emptyset$
	}
	
	Consider a vertex $v \in N(Q)$;
	
	\Return \textsf{Enum-Conn-Sets}$(G,Q \cup \{v\},a - 1,b) \cup \textsf{Enum-Conn-Sets}(G-\{v\},Q,a,b-1)$
	
  \caption{\textsf{Enum-Conn-Sets}: Enumeration Algorithm for $(Q,a,b)$-connected sets}
  \label{alg:bound-pc}
\end{algorithm}
\fi

\begin{lemma1} \label{lemma:main-no-conn-comps} For a graph $G$, a non-empty set $Q \subseteq V(G)$, and integers $a,b \in \mathbb{N}$, the number of $(Q,a,b)$-connected sets in $G$ is at most $2^{a + b - |Q|}$. Moreover, we can enumerate all $(Q,a,b)$-connected sets in $G$ in time $2^{a + b - |Q|}\cdot n^{\calO(1)}$.
\end{lemma1}

\begin{lemma1} \label{lemma:no-conn-comps} For a graph $G$ and integers $a,b \in \mathbb{N}$ the number of $(a, b)$-connected sets in $G$ is at most $2^{a + b} \cdot n^{\calO(1)}$. Moreover, we can enumerate all such sets in $2^{a + b} \cdot n^{\calO(1)}$ time.
\end{lemma1}\iffalse
\begin{proof}
 % Note that every non empty $(a, b)$-connected set is $(\{v\}, a, b)$-connected sets for some vertex $v$ in $G$. Hence proof of this lemma 
  Follows from Lemma~\ref{lemma:main-no-conn-comps}.
\end{proof}\fi

%% file: prob-disjoint-conn-3.tex
\section{\textsc{3-Disjoint Connected Subgraph}}
\label{sec:prob-dis-conn-3}

In this section, we define a generalization of \textsc{2-Disjoint Connected Subgraphs (2-DCS)}, called \textsc{3-Disjoint Connected Subgraphs (3-DCS)}. We design an algorithm for \textsc{$3$-DCS} running in time $\calO^{\star}(1.88^n)$, where $n$ is number of vertices in input graph. This algorithm will be useful in designing our algorithm for \textsc{Path Contraction}. %we use this algorithm to solve \textsc{$P_5$-Contraction}. 

In the following, we formally define the problem \textsc{$2$-DCS} which is studied in~\cite{cygan2014solving,telle2013connecting}. 

\defproblem{\textsc{2-Disjoint Connected Subgraphs (2-DCS)}}{A connected graph $G$ and two disjoint sets $Z_1$ and $Z_2$.}{Does there exist a partition $(V_1, V_2)$ of $V(G)$, such that for each $i\in [2]$, $Z_i\subseteq V_i$ and $G[V_i]$ is connected?
%  \begin{enumerate}[noitemsep,nolistsep]
% % \item Tuple $(V_1, V_2)$ is a partition of $V(G)$.
%  \item Sets $V_1, V_2$ are supersets of $Z_1, Z_2$, respectively.
%  \item Graphs $G[V_1]$ and $G[V_2]$ are connected.
%  \end{enumerate}
}

In the following we state a result regarding {\sc $2$-DCS} which will be useful later sections. 

\begin{proposition}[\cite{telle2013connecting} Theorem~3] \label{prop:exact-2-con} There exists an algorithm that solves \textsc{2-Disjoint Connected Subgraphs} problem in $\mathcal{O}^{\star}(1.7804^n)$ time where $n$ is number of vertices in the input graph.
\end{proposition}

In the \textsc{$3$-DCS} problem, the input is same as that of \textsc{2-DCS}, but we are interested in a partition of $V(G)$ into three sets, rather than two. We formally define the problem below. %The third part separates two subgraphs containing terminal sets. We formally define it as follows.

\defproblem{\textsc{3-Disjoint Connected Subgraphs (3-DCS)}}{A connected graph $G$ and two disjoint sets $Z_1$ and $Z_2$.}{Does there exist a partition $(V_1, U, V_2)$ of $V(G)$, such that 1) for each $i\in [2]$, $Z_i\subseteq V_i$ and $G[V_i]$ is connected, 2) $G[U]$ is connected, and 3) $G - U$ has exactly two connected components, namely, $G[V_1]$ and $G[V_2]$? 
%  \begin{enumerate}[noitemsep,nolistsep]
 % \item Tuple $(V_1, U, V_2)$ is a partition of $V(G)$.
  %\item Sets $V_1, V_2$ are supersets of $Z_1, Z_2$, respectively.
  %\item Graphs $G[V_1], G[V_2]$ and $G[U]$ are connected.
%  \item Graph $G - U$ has exactly two connected components viz $V_1, V_2$.
%  \end{enumerate}
}

We note that the problem definitions for \textsc{2-DCS} and \textsc{3-DCS} do not require the sets $Z_1, Z_2$ to be non-empty. If either of this set is empty, we can guess a vertex for each of the non-empty sets. Since there are at most $n^2$ such guesses, it will not affect exponential factor in the running time of our algorithm. Thus, here after we assume that both $Z_1$ and $Z_2$ are non-empty sets.

To design our algorithm for \TDCS, we first design an algorithm for a special case for the problem where the size of $Z_1\cup Z_2$ is at most $\delta n$, where $\delta= 0.092$. (The choice of $\delta$ will be clear when we present the proof.) We call the special case of \TDCS\ where $|Z_1\cup Z_2| \leq \delta n$, as \STDCS. In Section~\ref{subsec:stdcs}, we will design an algorithm for \STDCS\ running in time $\mathcal{O}^{\star}(1.88^n)$. That is, our goal of Section~\ref{subsec:stdcs} will be to prove the following lemma. 

\begin{lemma} \label{lemma:exact-small-terminal} 
\STDCS\ admits an algorithm running in time $\mathcal{O}^{\star}(1.88^n)$, where $n$ is the number of vertices in the input graph.
\end{lemma}

In the rest of this section, we show how we can use the above lemma to design an algorithm for the problem \TDCS, running in time $\mathcal{O}^{\star}(1.88^n)$. We also show how we can obtain an algorithm for {\textsc{$P_5$-Contraction}} running in time $\mathcal{O}^{\star}(1.88^n)$, using our algorithm for \TDCS. 

In the following theorem, we give our algorithm for \TDCS, using Lemma~\ref{lemma:exact-small-terminal} as a subroutine. 

\begin{theorem}\label{thm:exact-3-con} \TDCS\ admits an algorithm running in time $\mathcal{O}^{\star}(1.88^n)$, where $n$ is number of vertices in the input graph.
\end{theorem}
\begin{proof} Let $(G, Z_1, Z_2)$ be an instance of \TDCS. We consider the following two cases based on whether or not $|Z_1\cup Z_2| \leq \delta n$, where $ \delta = 0.092$. If $|Z_1\cup Z_2| \leq \delta n$, then we resolve the instance in time $\mathcal{O}^{\star}(1.88^n)$, using Lemma~\ref{lemma:exact-small-terminal}. Now we consider the case when $|Z_1\cup Z_2| > \delta n$. The goal is to look for a solution $(V_1,U,V_2)$ for the instance. We begin by enumerating all potential candidates for the set $U$. That is, we compute the set $\mcal{U} =\{U' \mid U' \subseteq V(G) \setminus (Z_1\cup Z_2)\}$. As $|V(G) \setminus (Z_1\cup Z_2)| \leq (1-\delta)n \leq 0.908n$, the time required to compute $\mcal{U}$ is bounded by $\mathcal{O}^{\star}(2^{0.908n}) \in \mathcal{O}^{\star}(1.88^n)$. Now for each $U' \in \mcal{U}$, we check the following properties: 1) $G[U']$ is connected and 2) $G-U'$ has exactly two connected components, one containing all vertices from $Z_1$ and the other containing all the vertices from $Z_2$. If there is $U'\in \mcal{U}$ which satisfies the above two conditions then we return that the instance is a yes-instance, and otherwise we return that the instance a no-instance. The correctness of the algorithm and the analysis of the claimed running time bound are apparent from the description. 
\end{proof}

Using Theorem~\ref{thm:exact-3-con} we obtain our algorithm for \textsc{$P_5$-Contraction}, in the following lemma.  

\begin{lemma} 
\PFC\ admits an algorithm running in time $\mathcal{O}^{\star}(1.88^n)$, where $n$ is the number of vertices in the input graph.
\end{lemma}
\begin{proof}
Let $G$ be a graph. By Observation~\ref{obs:singleton-end-bags}, if $G$ is contractible to $P_5$, then there exists a $P_5$-witness structure $\mathcal{W} = (W_1, \cdots, W_5)$ of $G$ such that $W_1$ and $W_5$ are singleton sets. We guess the pair of vertices which are in the sets $W_1$ and $W_5$, respectively. Note that there are at most $\calO(n^2)$ choices for such pairs. Let $W_1=\{x\}$ and $W_5= \{y\}$ be the current guess of these sets. If there is witness structure where $W_1=\{x\}$ and $W_5= \{y\}$, then the vertices in $N(x)$ and $N(y)$ must belong to the sets $W(t_2)$ and $W(t_4)$, respectively. (As otherwise, the contracted graph cannot be an (induced) path on $5$ vertices.) Note that with the above guess, the problem boils down to solving \TDCS\ on the instance $(G - \{x, y\}, N_G(x), N_G(y))$. Thus, we can use Theorem~\ref{thm:exact-3-con} to resolve the instance $(G - \{x, y\}, N_G(x), N_G(y))$ of \TDCS\ in time $\mathcal{O}^{\star}(1.88^n)$. This concludes the proof. 
\end{proof}

%\subsection*{\textsc{$P_5$-Contraction}}

%We use Theorem~\ref{thm:exact-3-con} to check whether a given graph $G$ can be contracted to $P_5$ or not. By Observation~\ref{obs:singleton-end-bags}, if $G$ is contractible to $P_5$ then there exists a $P_5$-witness structure $\mathcal{W} = (W_1, \cdots, W_5)$ of $G$ such that $W_1$ and $W_5$ are singleton sets. We guess the pair of vertices which are in the sets $W_1$ and $W_5$, respectively. Note that there are at most $\calO(n^2)$ choices for such pairs. Let $W_1=\{x\}$ and $W_5= \{y\}$ be the current guess of these sets. If there is witness set where $W_1=\{x\}$ and $W_5= \{y\}$, then the vertices in $N(x)$ and $N(y)$ must belong to the sets $W(t_2)$ and $W(t_4)$, respectively. (As otherwise, the contracted graph cannot be an (induced) path on $5$ vertices.) Note that with the above guess, the problem boils down to solving \TDCS\ on the instance $(G - \{x, y\}, N_G(x), N_G(y))$. From the above discussions (together with Theorem~\ref{thm:exact-3-con}), we obtain the following result. 

 %In graph $G - \{x, y\}$, we use $N(x), N(y)$ as set of terminals to find a tri-partition $(V_1, U, V_2)$. If exists, these three sets can work as witness structures corresponding to $W_2, W_3, W_4$ respectively. This simple algorithm implies following corollary of Theorem~\ref{thm:exact-3-con}.

\subsection{Algorithm for \STDCS}\label{subsec:stdcs}
The goal of this section will be to obtain a proof of Lemma~\ref{lemma:exact-small-terminal}, i.e., to design an algorithm for \STDCS\ running in time $\mathcal{O}^{\star}(1.88^n)$. Let $(G,Z_1,Z_2)$ be an instance of \STDCS. Note that $|Z_1\cup Z_2| \leq \delta n$, where $\delta = 0.092$. 

The intuition behind our algorithm is the following. We start by showing the existence of a special type of a solution, which we call an \emph{immovable tri-partition}, for a yes-instance. Roughly speaking, we use the properties ensured by a special solution to enumerate ``connectors'' for the set $Z_1\cup Z_2$ in an auxiliary graph. To enumerate such ``connectors'', we employ the algorithm of Telle and Villanger~\cite{telle2013connecting}. Then we show how we use these potential $Z_1\cup Z_2$ ``connectors'' in an auxiliary graph, to resolve the instance. 

In the following, we introduce some notations and preliminary results that will be useful in designing our algorithm. 

\paragraph*{Notations and Preliminary Results}
Consider a graph $H$ and set $Z\subseteq V(H)$. A vertex $v\in V(H)$ is called a \emph{$Z$-separator} if $Z$ contains vertices from at least two connected components of $G - \{v\}$. A set $S\subseteq V(H)$ is a \emph{$Z$-connector} if $Z\subseteq S$ and $H[S]$ is connected. Moreover, if no strict subset of $S$ is a $Z$-connector, then $S$ is a \emph{minimal $Z$-connector}. 

We state a result regarding enumeration of minimal $Z$-connectors in a graph which will be used in our algorithm. 

\begin{proposition}[\cite{telle2013connecting}]\label{prop:number-minimal-connector} Consider a graph $H$ on $n$ vertices and a set $Z \subset V(H)$ with at most $n/3$ vertices. Then, the number of minimal $Z$-connectors in $H$ is bounded by $\binom{n - |Z|}{|Z| - 2} \cdot 3^{(n - |Z|)/3}$. Moreover, we can enumerate all such minimal $Z$-connectors in time $\mathcal{O}^{\star}(\binom{n - |Z|}{|Z| - 2} \cdot 3^{(n - |Z|)/3})$. 
 %there are at most $\binom{n - |T|}{|T| - 2} \cdot 3^{(n - |T|)/3}$ minimal $T$-connecting vertex sets and those can be enumerated in time $\mathcal{O}^{\star}(\binom{n - |T|}{|T| - 2} \cdot 3^{(n - |T|)/3})$.
\end{proposition} 

In the following remark we state a criterion when we can directly conclude that the instance is a no-instance of \STDCS. The correctness of this remark will easily follow from the problem definition. (Henceforth we shall assume that the premise of the remark does not hold.)

\begin{remark}\label{rem:declare-no}
If there is an edge between $Z_1$ and $Z_2$, then conclude that $(G,Z_1,Z_2)$ is a no-instance of \STDCS. 
\end{remark}

A partition of $V(G)$ into three sets $(V_1,U,V_2)$ is a \emph{solution tri-partition} if the following conditions are satisfied: 1) for $i\in [2]$, $G[V_i]$ is connected and $Z_i\subseteq V_i$, 2) $G[U]$ is connected, and 3) $G-U$ has exactly two connected components, namely, $G[V_1]$ and $G[V_2]$. 

Now we will define a ``special solution'', which will be called an \emph{immovable tri-partition}, and we will show that if there is a solution, then there is also an immovable tri-partition. Our goal will be to find an immovable tri-partition, if it exists. Roughly speaking, an immovable tri-partition is a solution in which no vertex from $V_1 \cup V_2$ can be ``moved'' to the set $U$. 

\begin{definition}[Immovable tri-partition]\label{def:Immovable-Tri-partition} {\rm A solution tri-partition $(V_1, U, V_2)$ for $(G,Z_1,Z_2)$ of {\sc 3-DCS} is an \emph{immovable tri-partition}, if for every $i\in [2]$ and $v \in (V_i \setminus Z_i)\cap N(U)$ is a $Z_i$-separator in $G[V_i]$.}
\end{definition}

% Next, we formally define immovable tri-partitions. %Note that any vertex outside $Z$ which is not a cut vertex can not be a $Z$-separator.

In the following claim we show that an immovable tri-partition exists for a yes-instance.

\begin{claim}\label{claim:add-prop} If $(G, Z_1, Z_2)$ is a yes-instance of \textsc{$3$-DCS}, then there is an immovable tri-partition.  
\end{claim}
\begin{proof} Let $(V_1, U, V_2)$ be a solution tri-partition of $V(G)$. If this is an immovable tri-partition then we are done. Otherwise, assume that there is $v\in (V_1 \setminus Z_1)\cap N(U)$, such that $v$ is not a $Z_1$-separator in $G[V_1]$. (The case when there is $v \in (V_2 \setminus Z_2)\cap N(U)$, such that $v$ is not a $Z_2$-separator in $G[V_2]$ can be handled analogously.) Let $C_1, C_2, \cdots, C_d$ be the connected components of $G[V_1] - v$, where $d\geq 1$. Since $v$ is not a $Z_1$-separator, we know that $Z_1$ is contained in one of the connected components. Let $C_1$ be the connected component which contains $Z_1$. Consider the tri-partition $(V_1^{\prime}, U^{\prime}, V_2)$ of $V(G)$ where $V_1^{\prime} = V(C_1) = V_1 \setminus (\{v\} \cup V(C_2) \cup \cdots \cup V(C_d))$ and $U^{\prime} =  U \cup \{v\} \cup C_2 \cup \cdots \cup C_d$. This tri-partition is also a solution partition as both $V_1^\prime = C_1$ and $U^\prime$ are connected and $V_1^\prime$ contains $Z_1$. Following the above procedure, for a given tri-partition we can either find a vertex to move from $V_1 \cup V_2$ to $U$ or conclude that it is an immovable tri-partition. %Since every step reduces the number of vertices in $V_1 \cup V_2$ and we never add any vertex in $V_1 \cup V_2$ this process terminates in at most $n$ steps and returns an immovable tri-partition. 
\end{proof}

%Next we define the notion of $Z$-connectors in a graph. 

In the following claims we establish some useful properties regarding immovable tri-partitions. 

\begin{claim}\label{claim:connector-to-partition} Let $(G, Z_1, Z_2)$ be a \yes-instance of \textsc{$3$-DCS} and $(V_1, U, V_2)$ be an immovable tri-partition. Furthermore, let $S_1$ be a minimal $Z_1$-connector in $G[V_1]$ and $S_2$ be a minimal $Z_2$-connector in $G[V_2]$. Then, no connected component of $G[V_1] - S_1$ or $G[V_2] - S_2$ is adjacent to $U$.
\end{claim}
\begin{proof} Consider a connected component $C$ of $G[V_1] - S_1$. (The other case can be argued analogously.) As $S_1$ is $Z_1$-connector, $Z_1\subseteq S_1$ and $G[S_1]$ is connected. Since $(V_1, U, V_2)$ is an immovable tri-partition, no $v\in V(C)$ can be adjacent to a vertex in $U$, as $v$ is not a $Z_1$-separator (see Definition~\ref{def:Immovable-Tri-partition}). This concludes the proof. 
%Let $v$ be a vertex in $C$ which has neighbor in $U$. Note that since $S_1$ is a connected and $v$ is outside $S_1$, vertex $v$ is not $Z_1$-separator in $G[V_1]$. This contradicts the fact that $(V_1, U, V_2)$ is an immovable partition. Hence no vertex in any connected component of $G[V_1] - S_1$ has neighbors in $U$. Similar argument holds for any connected component of $G[V_2] - S$ leads to the same contradiction. Hence no connected component of $G[V_1] - S_1$ or $G[V_2] - S_2$ is adjacent with $U$.  
\end{proof}

%\begin{figure}[t]
%  \centering
%  \includegraphics[scale=0.75]{./images/minimal-connector.pdf}
%  \caption{Consider an instance $(G, Z_1, Z_2)$ of \TDCS\ with $Z_1 = \{z_1, z_3, z_5\}$ and $Z_2 = \{z_2, z_4\}$. Dotted line is the edge is added in $G$ to obtain $G'$. $(V_1, U, V_2)$ is an immovable tri-partition in $G$. $V_1$ and $V_2$ are minimal $Z_1$-connector and $Z_2$-connector in $G[V_1]$ and $G[V_2]$, respectively. $S = V_1 \cup V_2$ is minimal $(Z_1 \cup Z_2)$-connector in $G'$. %Please refer to Claim~\ref{claim:special-connector}. 
%  \label{fig:minimal-connector}}
%\end{figure}

As was mentioned earlier, we will construct an auxiliary graph, and relate a connector in the auxiliary graph to a solution for our instance. We now describe this auxiliary graph. Arbitrarily fix vertices $z_1\in Z_1$ and $z_2\in Z_2$. By $G'$, we denote the graph obtained from $G$ by adding the edge $z_1z_2$ to $G$. (From Remark~\ref{rem:declare-no} we know that there is no edge between $Z_1$ and $Z_2$ in $G$.) In the next claim we relate immovable tri-partitions of $G$ with minimal connectors in $G'$.

\begin{claim}\label{claim:special-connector} Let $(G, Z_1, Z_2)$ be a \yes-instance of \textsc{$3$-DCS} and let $(V_1,U,V_2)$ be an immovable tri-partition. Furthermore, let $S_1$ be a minimal $Z_1$-connector in $G[V_1]$ and $S_2$ be a minimal $Z_2$-connector in $G[V_2]$. Then, $S = S_1 \cup S_2$ is a minimal $(Z_1 \cup Z_2)$-connector in $G^{\prime}$.
\end{claim}
\begin{proof}
%Let $cc(G)$ denotes the number of connected components in graph $G$. Consider any two sets $X_1$ and $X_2$ which are subsets of $V_1$ and $V_2$ respectively. Since we have added only one edge between $V_1$ and $V_2$ while constructing graph $G'$, we get $cc(G'[X_1 \cup X_2]) \ge cc(G[X_1]) + cc(G[X_2]) - 1$.

We first argue that $S$ is a $(Z_1 \cup Z_2)$-connector in $G^{\prime}$. As $G^{\prime}[S_1]$ and $G^{\prime}[S_2]$ are connected and the edge $z_1z_2$ has one end point in $S_1$ and another in $S_2$, the graph $G^{\prime}[S]$ is connected. Since $S$ contains $Z_1 \cup Z_2$, it is a $(Z_1 \cup Z_2)$-connector. 

  It remains to argue that no proper subset of $S$ is a $(Z_1 \cup Z_2)$-connector. For the sake of contradiction, suppose that there is $v \in S$, such that $S'=S\setminus \{v\}$ is a $(Z_1 \cup Z_2)$-connector in $G^{\prime}$. We assume that $v\in S_1$. (The case when $v\in S_2$ can be argued symmetrically.) Let $S'_1 =S_1\setminus \{v\}$. Note that $G[S'_1]$ is not connected, as $S_1$ is a minimal $Z_1$-connector in $G[V_1]$. Let $C$ be a connected component of $G[S'_1]$ which does not contain $z_1$ (which exists as $G[S'_1]$ is not connected). Recall that $S'_1\subseteq V_1$, $S_2\subseteq V_2$, and $V_1\cap V_2=\emptyset$. Note that there can be no edge between $V(C)$ and $S_2$ in $G'$. This contradicts that $G'[S'_1 \cup S_2]$ is connected. This concludes the proof.  
%  As $S_2 \subseteq V_2$, we have $S_2\cap Z_1=\emptyset$. This implies that $G'[S]$ 
%  
%  As $S_2\subseteq S'$, all the vertices in $S_2$ is a $Z_2$-connector in $G[S_2]$. Also the edge $z_1z_2\in E(G')$. 
%  
%  Let $S'_1 = S' \cap V_1$ and $S'_2 = S' \cap V_2$. Consider a case when $S^{\prime}$ does not contain all vertices in $S_1 \setminus Z_1$. In other words, $S'_1$ is a proper subset of $S_1$. Recall that $S_1$ is a minimal $Z_1$-connector in $G[V_1]$ and $S^{\prime}_1$ contains $Z_1$. By minimality of $S_1$, graph $G[S ^{\prime}_1]$ is not connected and hence $cc(G[S^{\prime}_1]) \ge 2$. This implies $G'[S'] = G'[S'_1 \cup S'_2] \ge G[S^{\prime}_1] + G[S^{\prime}_1] - 1 \ge 2$ as $cc(G[S^{\prime}_2]) \ge 1$. This contradicts the fact that $G'[S']$ is a connected graph. By symmetric arguments, assuming $S^{\prime}_2$ is a proper subset of $S_2$ leads to same the contradiction. Hence our assumption is wrong and no proper subset of $S$ is a $(Z_1 \cup Z_2)$-connector.
\end{proof}

%We will see that we can compute an immovable tri-partitions in $G$ using ``minimal $Z$-connectors'' in the graph obtained by adding a specific edge in $G$. 

%We use terms \emph{$Z$-connecting} and \emph{$Z$-connector} interchangeably. 

%Note that above claim implies that in the graph $G - S_1$, there exists a unique connected component which is $G[U \cup V_2]$ and all other connected components of same as that of $G[V_1] - S_1$. 

%For the instance $(G, Z_1, Z_2)$, we assume that there is no edge with one end point in $Z_1$ and another in $Z_2$, as otherwise we can correctly conclude that it is a no-instance. 

We say that a minimal $(Z_1 \cup Z_2)$-connector $S$ in $G^{\prime}$ is \emph{realized} by an immovable tri-partition $(V_1, U, V_2)$ of $G$ if $S$ can be partitioned into two sets, $S_1, S_2$, such that $S_1$ is a minimal $Z_1$-connector in $G[V_1]$ and $S_2$ is a minimal $Z_2$-connector in $G[V_2]$. Note that Claim~\ref{claim:special-connector} implies that every immovable tri-partition of $V(G)$ realizes at least one  minimal $(Z_1 \cup Z_2)$-connector in $G'$. 
%\begin{algorithm}[t]
%  \KwIn{A graph $G$ and two disjoint non-empty sets $Z_1$ and $Z_2$.}
%  \KwOut{\yes\ or \no.}
%  Compute the set of relevant $Z$-connectors, $\mcal{S}_{\sf rel} \subseteq \mcal{S}$, in $G'$\;
%  \For{each $S\in \mathbb{S}$ }{
%    $\calW \leftarrow$ witness structure obtained by consider each connected component of $G[S]$ and $G - S$ as a witness set\;
%    $G' \leftarrow$ graph obtained from $G$ by contracting witness sets in $\calW$\;
%    \If{$G'$ is a path}{
%      $t = \max\{t, \text{ length of path } G'\}$\;
%    }
%  }
%  \Return $t$\;
%  \caption{Algorithm for \STDCS.}
%  \label{alg:stdcs}
%\end{algorithm}

\paragraph*{The Algorithm}
We are now ready to design our algorithm for \STDCS. Recall that $|Z_1 \cup Z_2| \le \delta n =0.092n$ and there is no edge between $Z_1$ and $Z_2$ (see Remark~\ref{rem:declare-no}). Let $Z=Z_1\cup Z_2$. Recall that $G'$ is the graph obtained from $G$ by adding the edge $z_1z_2$. Compute the set $\mcal{S}$ of all minimal $Z$-connectors in $G'$ using Proposition~\ref{prop:number-minimal-connector}. (The premise of the proposition is satisfied as $|Z| \leq 0.092n \leq n/3$.) We construct a set $\mcal{S}_{\sf rel} \subseteq \mcal{S}$ of \emph{relevant sets} as follows. Let $\mcal{S}_{\sf rel} = \{S \in \mcal{S} \mid G[S] \text{ has exactly two connected components } G[S_1] \text{ and } G[S_2], \text{ such that } Z_1 \subseteq S_1 \allowbreak \text{ and } Z_2 \subseteq S_2 \}$. 

Consider $S\in \mcal{S}_{\sf rel}$. Let $\mcal{C}_S$ be the set of connected components in $G-S$. Let $\mcal{C}_S^{\sf bth}$ be the set of components in $G-S$ that have a neighbor both in $S_1$ and $S_2$. That is, $\mcal{C}_S^{\sf bth}= \{C\in \mcal{C}_S \mid N(C) \cap S_1 \neq \emptyset \mbox{ and } N(C) \cap S_2 \neq \emptyset\}$.

Our algorithm does the following. If there is $S\in \mcal{S}_{\sf rel}$, such that $|\mcal{C}_S^{\sf bth}| = 1$, then the algorithm return that $(G,Z_1,Z_2)$ is a yes-instance of \STDCS. Otherwise it returns that the instance is a no-instance. 

In the following lemma we prove the correctness of the algorithm. 

\begin{figure}[t]
  \centering
  \includegraphics[scale=0.6]{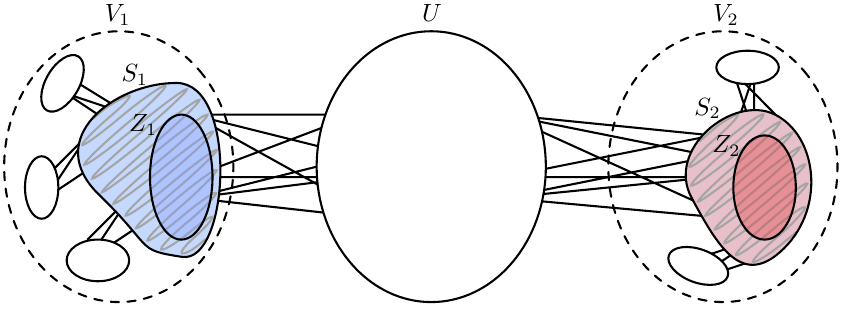}
  \caption{An illustration of various sets in the proof of forward direction of Lemma~\ref{lem:algo-correct-STDCS}. %Please refer to Claim~\ref{claim:special-connector}. 
  \label{fig:stdcs1}}
\end{figure}

\begin{lemma}\label{lem:algo-correct-STDCS}
The algorithm presented for \STDCS\ is correct.
\end{lemma}
\begin{proof}
In the forward direction, let $(G,Z_1,Z_2)$ be a yes-instance of \STDCS. We will show that there is $S\in \mcal{S}_{\sf rel}$, such that $|\mcal{C}_S^{\sf bth}| = 1$. Consider an immovable tri-partition $(V_1,U,V_2)$ for the instance (its existence is guaranteed by Claim~\ref{claim:add-prop}). Note that for $i\in [2]$, $Z_i \subseteq V_i$ and $G[V_i]$ is connected. Furthermore, $G[U]$ is connected and there are exactly two connected components in $G-U$, namely, $G[V_1]$ and $G[V_2]$.  
For $i\in [2]$, as $G[V_i]$ is connected, there is a minimal $Z_i$-connector $S_i$, in $G[V_i]$. Let $S=S_1\cup S_2$ (see Figure~\ref{fig:stdcs1}). From Claim~\ref{claim:special-connector}, $S$ is a minimal $Z$-connector in $G'$, where $Z=Z_1\cup Z_2$. Thus, $S \in \mcal{S}$. Note that $G[S]$ has exactly two connected components, namely, $G[S_1]$ and $G[S_2]$, and thus, $S\in \mcal{S}_{\sf rel}$. Recall $\mcal{C}_S$ denotes the set of connected components in $G-S$ and $\mcal{C}_S^{\sf bth}$ denotes the set of connected components in $G-S$ which have neighbors both in $S_1$ and $S_2$. To obtain the proof we will show that $\mcal{C}_S^{\sf bth} = \{G[U]\}$. Since $G[U]$ is connected and $S\cap U=\emptyset$, there is a component $C\in \mcal{C}_S$, such that $U\subseteq V(C)$. We first show that $V(C)\setminus U=\emptyset$. Towards a contradiction, assume that $V(C)\setminus U\neq \emptyset$. Then $V(C) \cap (V_1\cup V_2) \neq \emptyset$. Suppose that $V(C) \cap V_1 \neq\emptyset$. (The case when $V(C) \cap V_2 \neq\emptyset$ can be argued analogously.) Then there is a vertex $v \in V(C) \cap V_1$, such that $v\in (V_1\setminus Z_1)\cap N(U)$. From the above we can contradict the fact that $(V_1,U,V_2)$ is an immovable-tripartition. Thus we conclude that $V(C)=U$. Note that $U$ is adjacent to both $V_1$ and $V_2$ in $G$ and $G[U]$ is a connected component of $G-S$. Hence, $\emptyset \subset N(U) \cap V_1\subseteq S_1$ and $\emptyset \subset N(U)\cap V_2\subseteq S_2$. Thus, $G[U]\in \mcal{C}^{\sf bth}_S$. As no vertex in $V_1\cup V_2$ can be adjacent to both $S_1$ and $S_2$, we conclude that $\{G[U] \} = \mcal{C}^{\sf bth}_S$. 

In the reverse direction, assume that there is $S\in \mcal{S}_{\sf rel}$, such that $|\mcal{C}_S^{\sf bth}|=1$. We will construct a solution $(V_1,U,V_2)$ for the instance $(G,Z_1,Z_2)$, and hence establish that the instance is a yes-instance of \STDCS. Let $C^*$ be the unique connected component in $\mcal{C}_S^{\sf bth}$. As $S\in \mcal{S}_{\sf rel}$, $G[S]$ has exactly two connected components, $C_1$ and $C_2$, such that $Z_1\subseteq V(C_1)$ and $Z_2\subseteq V(C_2)$. For $i\in [2]$, let $\mcal{C}_i$ be the set of connected components different from $C^*$ that have a neighbor in $S_i$. Note that $\mcal{C}_1 \cup \mcal{C}_2 \cup \{C^*\} = \mcal{C}_S$ (and they are pairwise disjoint). Let $V_1= V(C_1) \cup (\bigcup_{C\in \mcal{C}_1} V(C))$, $V_2= V(C_2) \cup (\bigcup_{C\in \mcal{C}_2} V(C))$, and $U=V(C^*)$. Note that $(V_1,U,V_2)$ is a partition of $V(G)$, for each $i\in [2]$, $G[V_i]$ is connected and $Z_i\subseteq V_i$, $G[U]$ is connected, and $G-U$ has exactly two connected components namely, $G[V_1]$ and $G[V_2]$. (In the above we rely on the connectedness of $G$.) Hence, $(G,Z_1,Z_2)$ is a yes-instance of \STDCS. 
\end{proof}

\begin{proof}[Proof of Lemma~\ref{lemma:exact-small-terminal}] From Lemma~\ref{lem:algo-correct-STDCS}, we know that the algorithm presented for \STDCS\ is correct. Thus to establish the proof of the lemma, it is enough to argue that the algorithm presented for \STDCS\ runs in time $\calO^{\star}(1.88^n)$. The only step of the algorithm that requires exponential time is the computation of set $\mcal{S}$ of all minimal $Z$-connectors in $G'$ where we use Proposition~\ref{prop:number-minimal-connector}. As $|Z|=|Z_1\cup Z_2| \leq \delta n=0.092n$, the time required to compute $\mcal{S}$ is bounded by $\mathcal{O}^{\star}(\binom{n - |Z_1 \cup Z_2|}{|Z_1 \cup Z_2| - 2} \cdot 3^{(n - |Z_1 \cup Z_2|)/3})$, which is bounded by $\calO^{\star}(\binom{(1 - \delta)n}{\delta n} \cdot 3^{(1 - \delta)n/3})$. Using Observation~\ref{obs:subset-no}, we can bound the above by $\calO^{\star}((\frac{(1 - \delta)^{(1 - \delta)}}{\delta^{\delta}\cdot {(1 - 2\delta)}^{(1 - 2\delta)}})^n)$. Using a computer program we obtained that $\delta \sim 0.092$ would be the value for which the overall running time of \textsc{3-DCS} is optimized. Thus, the running time of our algorithm is bounded by $\calO^{\star}(1.88^n)$. 
\end{proof}

%% file: exact_algo_path.tex
\section{Exact Algorithm for {\sc Path Contraction}}
\label{sec:exact-algo}
In this section we design our algorithm for {\sc Path Contraction} running in time $\mathcal{O}^{\star}(1.99987^n)$, where $n$ is the number of vertices in the input graph. To design our algorithm, we design four different subroutines each solving the problem {\sc Path Contraction}. Each of these subroutines is better than the other when a specific ``type'' of solution exists for the input instance. Thus the main algorithm will use these subroutines to search for solutions of the type they are the best for. We also design a subroutine for enumerating special types of partial solutions, which will be used in some of our algorithms for {\sc Path Contraction}. 

In the following we briefly explain the four subroutines and describe when they are useful. Let $G$ be an instance for {\sc Path Contraction}, where $G$ is a graph on $n$ vertices. Let $t$ be the largest integer (which we do not know a priori), such that $G$ is contractible to $P_t$ with $(W_1,W_2,\cdots, W_t)$ as a $P_t$-witness structure of $G$. We let $\os$ and $\es$ be the union of vertices in odd and even witness sets, respectively. That is, $\os = \bigcup_{x=1}^{\lceil t/2 \rceil} W_{2x - 1} \text{ and } \es = \bigcup_{x=1}^{\lfloor t/2 \rfloor} W_{2x}$. 
%Roughly speaking, the sets $\os$ and $\es$ are particularly useful when their sizes are small, as in that case we can enumerate over all such ``small'' sets. Moreover, once we remove the correct guess for the set $\os/\es$ (whichever is small), in the remaining graph, each connected component must be contracted to a single vertex, and the connected components in the ``guessed set'' must be contracted to a single vertex. Two of our subroutines are ``almost'' centered around exploiting this property. The other two subroutines will be centered around ``partitioning'' of the witness structure ``along'' a witness set.  

We now give an intuitive idea of the purposes of each of our subroutines in the main algorithm, while deferring their implementations to the subsequent sections. We also describe a subroutine which will help us build ``partial solutions'', and this subroutine will be used in two of our subroutines for {\sc Path Contraction}. (We refer the reader to Figure~\ref{fig:algo-subroutine} for an illustration of it.) 

\begin{figure}[t]
  \centering
  \includegraphics[scale=0.5]{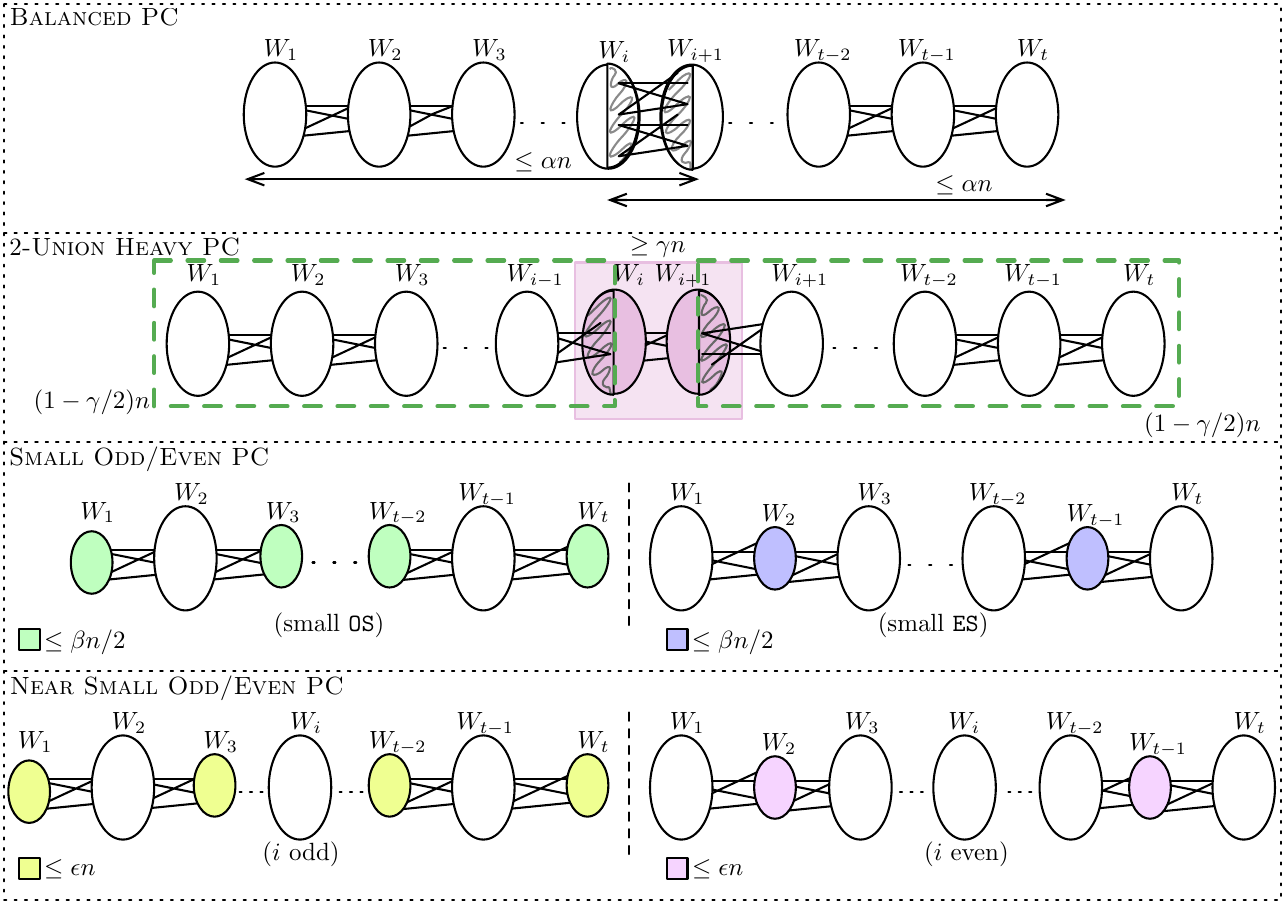}
  \caption{Various subroutines for the algorithm and their usage.}
  \label{fig:algo-subroutine}
\end{figure}

\subparagraph*{\BPC} This subroutine is useful when we can ``break'' the graph into two parts after a witness set, such that the closed neighborhood for each of the parts has small size, or in other words, the parts are ``balanced''. The quantification of the ``balancedness'' after a witness set will be done with the help of a rational number $0< \alpha \leq 1$, which will be part of the input for the subroutine. The subroutine will only look for those $P_t$-witness structures for $G$ for which there is an integer $i\in [t]$, such that the sizes of both $N[\bigcup_{j\in [i]}W_j]$ and $N[\bigcup_{j\in [t]\setminus [i]}W_j]$ are bounded by $\alpha n$. Moreover, the algorithm will return the largest such $t$. Our algorithm for \BPC\ will run in time $\mcal{O}^\star(2^{\alpha n})$. Note that when $\alpha=1$, \BPC\ is an algorithm for {\sc Path Contraction} running in time $\mcal{O}^\star(2^n)$.

\subparagraph*{\TDCPC} This subroutine will be used when a ``large'' part of the graph is concentrated in two consecutive witness sets and the neighborhood of the rest of the graph into them is ``small''. The quantification of term ``large/small'' will be done by a a fraction $0< \gamma <1$, which will be part of the input. The algorithm will only search for those $P_t$-witness structure of $G$ where there is an integer $i\in [t-1]$, such that $|W_i\cup W_{i+1}| \geq \gamma n$, and $|N[\bigcup_{j\in [i-1]}W_j]|, |N[\bigcup_{j\in [t]\setminus [i+1]}W_j]| \leq (1-\gamma/2) n$. Moreover, the algorithm will return largest such $t$. 

% are at least $\gamma n/2$. The value of gamma will be provided as input. This subroutine uses \BPC\ and creates instances of \textsc{2-Disjoint Connected Subgraphs}. The instances of \textsc{2-Disjoint Connected Subgraphs} are resolved using the algorithm of Telle and Villanger~\cite{telle2013connecting}. The running time of \TDCPC\ will be bounded by $\mcal{O}^\star(2^{1-\gamma/2}+c^n)$, where $c=\max_{\gamma\leq \delta \leq 1}\{1.7804^\delta \cdot g(1-\delta)^n\}$. The subroutine will return the largest integer $t\geq 1$, for which there is a $P_t$-witness structure for $G$, for which there is $i\in [t-1]$, such that the sizes of both $W_i$ and $W_{i+1}$ are at least $\gamma n/2$. Note that for $\gamma=2/n$, \TDCPC\ is an algorithm for {\sc Path Contraction}. 

\subparagraph*{\SOEPC} Roughly speaking, this subroutine is particularly useful when one of $\os$ or $\es$ is ``small''. The ``smallness'' of $\os/\es$ is quantified by a rational number $0< \beta \leq 1$, which will be part of the input. The subroutine will only look for those $P_t$-witness structures for $G$ where one of $|\os|\leq \beta n/2$ or $|\es|\leq \beta n/2$ holds. Moreover, the algorithm will return the largest integer $t\geq 1$, for which such a $P_t$-witness structure for $G$ exists. \SOEPC\ will run in time $\mcal{O}^\star(c^n)$, where $c=g(\beta/2)$. We note that when $\beta=1$, then one of $|\os|\leq \beta n/2$ or $|\es|\leq \beta n/2$ definitely holds. Thus, for $\beta=1$, \SOEPC\ is an algorithm for {\sc Path Contraction} running in time $\mcal{O}^\star(2^n)$ (see Observation~\ref{obs:subset-no}).  

\subparagraph*{\NSOEPC} In the case when both $\os$ and $\es$ are ``large'', it may be the case that for one of $\os/\es$, there is just one witness set which is large. That is, when we remove this large witness set, then one of $\os/\es$ becomes ``small''. The ``smallness'' of the remaining \os/\es\ (after removing a witness set) will be quantified by a rational number $0< \epsilon \leq 1$, which will be part of the input. The subroutine will only look for those $P_t$-witness structures for $G$ where the size of one of $|\os|$ or $|\es|$ after removal of a witness set is bounded by $\epsilon n$. Moreover, the algorithm will return the largest such $t$.

Our subroutines \BPC\ and \TDCPC\ use a subroutine called \EPPC\ for enumerating solutions for ``small'' subgraphs. The efficiency of the algorithm for \EPPC\ is centered around the bounds for $(Q,a,b)$-connected sets. In Section~\ref{sub-sec:eppc} we (define and) design an algorithm for \EPPC. In Sections~\ref{sub-sec:dp},~\ref{sub-sec:dis-conn},~\ref{sub-sec:enum} and~\ref{sub-sec:dis-conn-3} we present our algorithms for \BPC, \TDCPC, \SOEPC, and \NSOEPC, respectively. Finally, in Section~\ref{sub-sec:algo} we show how we can use the above algorithms to obtain an algorithm for {\sc Path Contraction}, running in time $\calO^\star(\runtime^n)$.

\input{enumerating-partialPC.tex}
\input{dynamic-programming.tex}

\input{disjoint-conn.tex}

\input{enumerating-subsets.tex}

\input{disjoint-conn-3.tex}

\input{algo.tex}

%% file: enumerating-partialPC.tex
\subsection{Algorithm for \EPPC}\label{sub-sec:eppc}

In this section, we describe an algorithm which computes a ``nice solution'' for all ``$\rho$-small'' subset of vertices of an input graph. In an input graph $G$, for a set $S\subseteq V(G)$, by $\Phi(S)$ we denote the set of vertices in $S$ that have a neighbor outside $S$. That is, $\Phi(S)= \{s\in S \mid N(s)\setminus S \neq\emptyset\}$. A set $S\subseteq V(G)$ is $\rho$-\emph{small} if $N[S] \leq \rho n$.
For an $\rho$-small set $S\subseteq V(G)$, the largest integer $t_S$ is called the \emph{nice solution} if $G[S]$ is contractible to $P_{t_S}$ with all the vertices in $\Phi(S)$ in the end bag. That is, there is a $P_{t_S}$-witness structure $(W_1,W_2,\cdots W_{t_S})$ of $G[S]$, such that $\Phi(S) \subseteq W_{t_S}$.
We formally define the problem \EPPC\ in the following way. 

\defproblemout{\EPPC}{A graph $G$ on $n$ vertices and a fraction $0< \rho \leq 1$.}{A table $\Gamma$ which is indexed by $\rho$-\emph{small} sets. For any $\rho$-\emph{small} set $S$,  $\Gamma[S]$ is the largest integer $t$ for which $G[S]$ has a $P_t$-witness structure $\mcal{W}=(W_1,W_2,$ $\cdots, W_t)$, such that $\Phi(S) \subseteq W_t$.
}
\vspace{0.2cm}

We design an algorithm for \EPPC\ running in time $\mcal{O}^\star(2^{\rho n})$.
We briefly explain how we can compute nice solutions for $\rho$-small set. Consider a $\rho$-small set $S$. Note that $|S| \leq \rho n$. Thus, by the method of $2$-coloring (as was explained in the introduction), we can obtain the nice solution in time $2^{\rho n}$. This would lead us to an algorithm running in time $\mcal{O}^\star(2^{\rho n}g(\rho)^n)$ (See Inequality~(\ref{eq:enumeration-over-subsets})). By doing a simple dynamic programming we can also obtain an algorithm running in time $\mcal{O}^\star(3^{n})$. We will improve upon these algorithms by a dynamic programming algorithm where we update the values ``forward'' instead of looking ``backward''. 

\iffalse
\begin{algorithm}[t]
  \KwIn{A graph $G$ and a fraction $0<\rho \leq 1$.}
  \KwOut{A table $\Gamma$ such that for every $\rho$-small set $S$, $\Gamma[S]$ is the largest integer $q$ for which $G[S]$ can be contracted to $P_q$ with $\Phi(S)$ is in the end bag.}
Compute $\mathcal{S} =\{S\subseteq V(G) \mid G[S] \mbox{ is connected and } |N[S]| \le \rho n\}$, using Lemma~\ref{lemma:no-conn-comps}\;% Set of all connected subset $S$ of $V(G)$ such that $|N[S]| \le \rho n$\; 
%  \tcc{ $\mathcal{T}[S]$ will store the maximum number of bags in path contraction of $G[S]$ which contains $\Phi(S)$ in the last bag.}
  \For{$S \in \mathcal{S}$}{
  Initialize $\Gamma[S] = 1$\;
    }

  \For{$S \in \mathcal{S}$ (in increasing order of their sizes)}{
% Let $x = |S|$;\  $y = |N(S)|$\;  
    \For{every pair $(a, b)$ of positive integers s.t. $|S| + a + b \le \rho n$ and $|N(S)| \leq a$}{
      Compute $\calA_{a, b}[S] = \{A\subseteq V(G-S) \mid G-S[A] \mbox{ is connected, } N_G(S)\subseteq A, |A|=a, \mbox{ and } |N_{G-S}(A)| = b\}$, using Lemma~\ref{lemma:main-no-conn-comps}\;
    %  \tcc{$\calA_{a, b}[S]$ is a collection of connected set $A$ in $G - S$ such that $N(S) \subseteq A$, $|A| = a$ and $|N(A) \setminus S| \le b$.}
        \For{$A \in \calA_{a, b}[S]$}{
          $\Gamma[S \cup A] = \max\{\Gamma[S \cup A], \Gamma[S] + 1\}$ \label{step:recurrence}\;        
          }
        }
    }
    
  \Return $\Gamma$\;
  \caption{Algorithm for \EPPC}
  \label{method:partial-EPPC}
\end{algorithm}
\fi

\begin{algorithm}
  \caption{Algorithm for \EPPC}
  \label{method:partial-EPPC}
  \begin{algorithmic}[1]
     \REQUIRE {{A graph $G$ and a fraction $0<\rho \leq 1$.}}
    \ENSURE {{A table $\Gamma$ such that for every $\rho$-small set $S$, $\Gamma[S]$ is the largest integer $q$ for which $G[S]$ can be contracted to $P_q$ with $\Phi(S)$ is in the end bag.}}
    %\STATE{\algorithmicrequire{ A graph $G$ and a fraction $0<\rho \leq 1$.}}
    %\STATE{\algorithmicensure{ A table $\Gamma$ such that for every $\rho$-small set $S$, $\Gamma[S]$ is the largest integer $q$ for which $G[S]$ can be contracted to $P_q$ with $\Phi(S)$ is in the end bag.}}
    \vspace{0.1cm}
    \STATE{Compute $\mathcal{S} =\{S\subseteq V(G) \mid G[S] \mbox{ is connected and } |N[S]| \le \rho n\}$ (Lemma~\ref{lemma:no-conn-comps}) \label{step:enum-rho-small}}
    \FOR{$S \in \mathcal{S}$ \label{step:for-start}}
    \STATE{Initialize $\Gamma[S] = 1$}
    \ENDFOR
    \FOR{$S \in \mathcal{S}$ (in increasing order of their sizes) \label{step:big-for-start}}
		\FOR{every pair $(a, b)$ of positive integers s.t. $|S| + a + b \le \rho n$ and $|N(S)| \leq a$}
    	\STATE{ Compute $\calA_{a, b}[S] = \{A\subseteq V(G-S) \mid G-S[A] \mbox{ is connected, } N_G(S)\subseteq A, |A|=a, \mbox{ and } |N_{G-S}(A)| = b\}$, using Lemma~\ref{lemma:main-no-conn-comps}}
		    \FOR{$A \in \calA_{a, b}[S]$}
		    \STATE{$\Gamma[S \cup A] = \max\{\Gamma[S \cup A], \Gamma[S] + 1\}$  \label{step:recurrence}}
    	    \ENDFOR
    	\ENDFOR
    \ENDFOR
	\RETURN{$\Gamma$}
  \end{algorithmic}
\end{algorithm}

\subparagraph{The Algorithm} 
We start by defining the tables entries for our dynamic programming routine, which is used for computation of nice solutions. (The pseudo code for our algorithm is presented in Algorithm~\ref{method:partial-EPPC}.) Let $\mathcal{S}$ be the set of connected $\rho$-small sets. 
That is, $\mathcal{S} =\{S\subseteq V(G) \mid G[S] \mbox{ is connected and } |N[S]| \le \rho n\}$. For each $S\in \mcal{S}$, there is an entry, denoted by $\Gamma[S]$, in the table which stores a nice solution for $S$.
In other words, $\Gamma[S]$ is the largest integer $q \geq 1$ for which $G[S]$ can be contracted to $P_q$ with a $P_q$-witness structure $\mcal{W}=(W_1,W_2,$ $\cdots, W_q)$ of $G[S]$, such that $\Phi(S) \subseteq W_q$. The algorithm starts by initializing $\Gamma[S]=1$, for each $S\in \calS$. 

In the following we introduce some notations that will be useful in stating the algorithm. Consider $S \in \mcal{S}$. We will define a set $\calA[S]$, which will be the set of all ``potential extenders bags'' for $S$, when we look at contraction to paths for larger graphs (containing $S$). For the sake of notational simplicity, we will define $\calA_{a,b}[S] \subseteq \calA[S]$, where the sets in $\calA_{a,b}[S]$ will be of size exactly $a$ and will have exactly $b$ neighbors outside $S$. We will define the above sets only for ``relevant'' $a$s and $b$s. We now move to the formal description of these sets. Consider $S\in \mcal{S}$ and integers $a,b$, such that $|S|+a+b \leq \rho n$ and $|N(S)| \leq b$. We let $\calA_{a, b}[S] = \{A\subseteq V(G-S) \mid G-S[A] \mbox{ is connected, } N_G(S)\subseteq A, |A|=a, \mbox{ and } |N_{G-S}(A)| = b\}$. 

The algorithm now computes nice solutions. The algorithm considers sets from $S\in \calS$, in increasing order of their sizes and does the following. (Two sets that have the same size can be considered in any order.) For every pair of integers $a,b$, such that $|S|+a+b \leq \rho n$ and $|N(S)| \leq b$, it computes the set $\calA_{a, b}[S]$. Note that $\calA_{a, b}[S]$ can be computed in time $\calO^\star(2^{a+b-|S|})$, using Lemma~\ref{lemma:main-no-conn-comps}. Now the algorithm considers $A \in \calA_{a, b}[S]$. Intuitively speaking, $A$ is the ''new'' witness set to be ``appended'' to the witness structure of $G[S]$, to obtain a witness structure for $G[S\cup A]$. Thus, the algorithm sets $\Gamma[S\cup A] = \max\{\Gamma[S\cup A], \Gamma[S]+1\}$. This finishes the description of our algorithm. 

In the following few lemmas we establish the correctness and runtime analysis of the algorithm. 

\begin{lemma}\label{lem:correct-EPPC}
For each $S\in \calS$, the algorithm computes $\Gamma[S]$ correctly. 
\end{lemma}
\begin{proof}
We prove the statement by induction on the size of sets in $\calS$. The base case is for sets of sizes $1$. That is, for the base case we show that for each $S\in \calS$, such that $|S|=1$, the algorithm computes $\Gamma[S]$ correctly. Consider a set $S\in \calS$ of size $1$. Note that in this case, $\Gamma[S]$ must be equal to $1$. At Step~\ref{step:for-start}, the algorithm initializes $\Gamma[S']=1$, for each $S'\in \calS$. Note that no other step of the algorithm modifies the value of $\Gamma[S]$ (as $|S|=1$). Thus, the algorithm correctly computes $\Gamma[S]$. 

For the induction hypothesis, we assume that the algorithm computes $\Gamma[S']$ correctly, for each $S'\in \calS$, such that $|S'| \leq r$. We will now argue that the computation of $\Gamma[\cdot]$ for sets of size $r+1$ are correct. Consider $S\in \calS$, such that $|S| = r+1$. Let $q_{\sf opt}$ be the nice solution for $S$ and $q_{\sf out}$ be the value of $\Gamma[S]$ computed by the algorithm. We will show that $q_{\sf out} = q_{\sf opt}$. 

Firstly, we show that $q_{\sf out} \geq  q_{\sf opt}$. Consider a $P_{q_{\sf opt}}$-witness structure $\calW = (W_1, W_2,\allowbreak \cdots, W_{q_{\sf opt}})$ of $G[S]$, such that $\Phi(S) \subseteq W_{q_{\sf opt}}$. Note that $q_{\sf out}\geq 1$, thus if $q_{\sf opt}=1$, then $q_{\sf out} \geq  q_{\sf opt}$ trivially holds. Now we consider the case when $q_{\sf opt} \geq 2$. Let $\what S = S\setminus W_{q_{\sf opt}}$, $a= |W_{q_{\sf opt}}|$, and $b = |N(W_{q_{\sf opt}})|$. As $S\in \calS$, we have $|N[S]| \leq \rho n$. Thus, $|\what S| +a + b \leq \rho n$. Since $\calW$ is a $P_{q_{\sf opt}}$-witness structure of $G[S]$, we have $N(\what S) \subseteq W_{q_{\sf opt}}$ and thus, $a \geq |N(\what S)|$. Also, $\emptyset \neq \what S \in \calW$. (In the above we rely on the fact that $q_{\sf opt} \geq 2$.) From the above discussions we can conclude that the set $\calA_{a,b}[\what S]$ is well defined, and $W_{q_{\sf opt}} \in \calA_{a,b}[\what S]$. By the induction hypothesis, we known that $\Gamma[\what S] \geq q_{\sf opt} -1$ is correctly computed. Thus, at Step~\ref{step:recurrence}, the algorithm sets $q_{\sf out}= \Gamma[S] \geq q_{\sf opt}$. Hence, we conclude that $q_{\sf out} \geq  q_{\sf opt}$.  

Next, we show that $q_{\sf out} \leq  q_{\sf opt}$. Note that $q_{\sf opt}\geq 1$. Thus, if $q_{\sf out}=1$, then the claim is trivially satisfied. We next consider the case when $q_{\sf out} \geq 2$. As $q_{\sf out} \geq 2$, there is a set $\what S \in \cal S$, with $\what S \subset S$ ($\what S\neq S$) and integers $a,b$, such that $|\what S| + a + b \leq \rho n$, $|N(\what S)| \leq a$, such that $A = S\setminus \what S \in \calA_{a,b}[\what S]$ and $q_{\sf out} = \Gamma[\what S] +1$. By induction hypothesis, $\Gamma[\what S]$ is computed correctly. Thus, there is a $P_{q_{\sf out}-1}$-witness structure $\calW'=(W_1, W_2, \cdots W_{q_{\sf out}-1})$ of $G[\what S]$, such that $\Phi(\what S) \subseteq W_{q_{\sf out}-1}$. But then, $\calW=(W_1, W_2, \cdots W_{q_{\sf out}-1}, A)$ is a $P_{\sf out}$-witness structure of $G[S]$. Thus, $q_{\sf out} \leq q_{\sf opt}$. This concludes the proof. 
\end{proof}

\begin{lemma}\label{lem:correct-runtime-EPPC}
\sloppy The algorithm presented for \EPPC\ runs in time $\calO^\star(2^{\rho n})$.  
\end{lemma}

\begin{proof}
Steps~\ref{step:enum-rho-small} and~\ref{step:for-start} of the algorithm can be executed in time bounded by $\calO^\star(2^{\rho n})$ (see  Lemma~\ref{lemma:no-conn-comps}). We will now argue about the time required for execution of \emph{for}-loop starting at Step~\ref{step:big-for-start}. Towards this, we start by partitioning sets in $\calS$ by their sizes and the sizes of their neighborhood. Recall that for any $S\in \calS$, we have $|N[S]| \leq \rho n$. Let $\calS_{x,y}=\{S\in \calS \mid |S|=x \mbox{ and } |N(S)|=y\}$, where $x,y \in [\lfloor \rho n \rfloor]$, such that $x+y \leq \rho n$. Consider $x,y \in [\lfloor \rho n \rfloor]$, where $x+y \leq \rho n$. From Lemma~\ref{lemma:no-conn-comps}, $|\calS_{x,y}|$ is bounded by $\calO^\star(2^{x+y})$. For each $S\in \calS_{x,y}$, the algorithm considers every pair of integers $a,b$, such that $|S|+a+b \leq \rho n$ and $|N(S)| \leq a$, and computes the set $\calA_{a, b}[S]$. Note that $\calA_{a, b}[S]$ can be computed in time bounded by $\calO^\star(2^{a+b-|N(S)|})$, from Lemma~\ref{lemma:main-no-conn-comps}. Furthermore, the algorithm spends time proportional to $|\calA_{a, b}[S]|$, at Step~\ref{step:recurrence}. From the above discussions, we can bound the running time of the algorithm by the following.
$$\calO^\star(
\sum_{\substack{x,y \\ x+y\leq \rho n}} 
%\sum_{S\in \calS_{x,y}}
\sum_{\substack{a,b \\ x+a+b \leq \rho n}} 
 2^{x+y} \cdot 2^{a+b-y} 
 ) = 
 \calO^\star(
\sum_{\substack{x,y \\ x+y\leq \rho n}} 
%\sum_{S\in \calS_{x,y}}
\sum_{\substack{a,b \\ x+a+b \leq \rho n}} 
 2^{x+a+b} 
 ) = \calO^\star(2^{\rho n})
 $$
This concludes the proof.
\end{proof}

%% file: dynamic-programming.tex
%\subsection{Method Using Dynamic Programming}
\subsection{Algorithm for \BPC}\label{sub-sec:dp}

We formally define the problem \BPC\ in the following. 

\defproblemout{\BPC}{A graph $G$ on $n$ vertices and a fraction $0<\alpha\leq 1$.}{Largest integer $t\geq 2$ for which $G$ has a $P_t$-witness structure $\mcal{W}=(W_1,W_2,$ $\cdots, W_t)$, such that there is $i\in [t]$ with $N[\bigcup_{j \in [i]}W_j] \leq \alpha n$ and $N[\bigcup_{j \in [t] \setminus [i]}W_j] \leq \alpha n$. Moreover, if no such $t$ exists, then output $1$.
}
\vspace{0.2cm}

%In the above, if their is no such integer $t>2$, then the output $2$. 
We design an algorithm for \BPC\ running in time $\mcal{O}^\star(2^{\alpha n})$. Let $(G,\alpha)$ be an instance of \BPC.

\begin{figure}[t]
  \centering
  \includegraphics[scale=0.75]{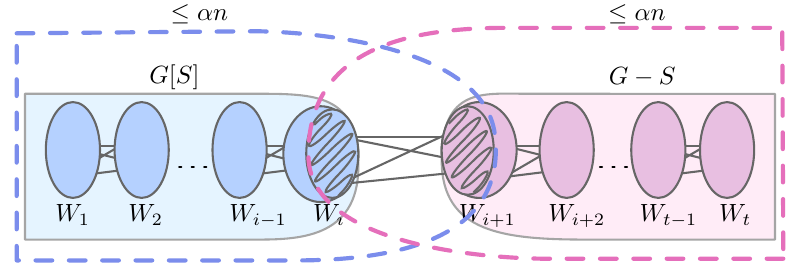}
  \caption{An illustration of construction of the solution using solutions for instances of smaller sizes.}\label{fig:intuition-dp}
\end{figure}

%We first present informal description of algorithm.
We begin by explaining the intuition behind the algorithm.
Recall that for an $\alpha$-small set $S\subseteq V(G)$, integer $t_S$ is called the \emph{nice solution} if $G[S]$ is contractible to $P_{t_S}$ with all the vertices in $\Phi(S)$ in the end bag. That is, there is a $P_{t_S}$-witness structure $(W_1,W_2,\cdots W_{t_S})$ of $G[S]$, such that $\Phi(S) \subseteq W_{t_S}$.
Suppose that we know the value of $t_S$ for every $\alpha$-small set $S$. 
Now we see how we can use these nice solutions for $\alpha$-small sets to solve our problem (see Figure~\ref{fig:intuition-dp}). Recall that we are looking for the largest integer $t$, such that $G$ is contractible to $P_t$, with $\mcal{W}=(W_1,W_2,$ $\cdots, W_t)$ as a $P_t$-witness structure of $G$, such that there is $i\in [t]$ with $|\bigcup_{j \in [i+1]}W_j| \leq \alpha n$ and $|\bigcup_{j \in [t] \setminus [i-1]}W_j| \leq \alpha n$. Let $S= \bigcup_{j \in [i]}W_j$. As $|\bigcup_{j \in [i+1]}W_j| \leq \alpha n$ and $N(S)\subseteq W_{i+1}$, the set $S$ is an $\alpha$-small set. Similarly, we can argue that $V(G)\setminus S$ is an $\alpha$-small set. Thus, for $S$ and $V(G)\setminus S$, we know the nice solutions $t_{S}$ and $t_{V(G)\setminus S}$, respectively. Notice that the solution to the whole graph is actually $t_S+t_{V(G)\setminus S}$.
%The intuition explained above is used in designing our algorithm.

%We begin by explaining the intuition behind the algorithm.
%We will compute ``nice solutions'' for each $\alpha$-small set using dynamic programming. Then using these ``nice solutions'' we will construct a solution for the whole instance.
%Before moving to intuition for computing ``nice solutions'' for $\alpha$-small sets, we briefly explain what we mean by these ``nice solutions'' and how we can use them to compute the solution for the whole instance. 
%Now we will describe how the algorithm uses nice solutions to compute the solution for the whole instance.
\subparagraph{The Algorithm}  The algorithm initializes $t=1$. (At the end, $t$ will be the output of the algorithm.)
The algorithm computes table $\Gamma = $\EPPC$(G, \alpha)$ using Algorithm~\ref{method:partial-EPPC}.
% Then using these ``nice solutions'', the algorithm construct a solution for the whole instance.
Let $\mathcal{S}$ be the set of all connected sets  $S$ in $G$, such that $|N[S]| \leq \alpha n$. 
That is, $\mathcal{S} =\{S\subseteq V(G) \mid G[S] \mbox{ is connected and } |N[S]| \le \alpha n\}$. For each $S\in \mcal{S}$, we have an entry denoted by $\Gamma[S]$.
The algorithm considers each $S\in \calS$ for which $V(G)\setminus S \in \calS$. It sets $t=\max\{t, \Gamma[S] + \Gamma[V(G)\setminus S]\}$. Finally, the algorithm returns $t$ as the output. This completes the description of the algorithm. 

\begin{lemma}\label{lem:correct-bpc}
The algorithm presented for \BPC\ is correct. 
\end{lemma}
\begin{proof}
For an instance $(G,\alpha)$, suppose that $t_{\sf opt}$ is the solution to \BPC\ and $t_{\sf out}$ is the output returned by the algorithm. We will show that $t_{\sf out} = t_{\sf opt}$. 

%and $t_{\sf out}$ are the solution to \BPC\ and the 
%Suppose that $G$ has a $P_t$-witness structure $\mcal{W}=(W_1,W_2,$ $\cdots, W_t)$, for which there is $i\in [t]$ with $N[\cup_{j \in [i]}W_j] \leq \alpha n$ and $N[\cup_{j \in [t] \setminus [i]}W_j] \leq \alpha n$. If $t=1$, then $\alpha$ must be equal to $1$. In the above case, the algorithm correctly returns $1$. Now suppose that 
%If $\alpha < 1$, then $t \geq 2$ and $t $. 

\sloppy Firstly, we show that $t_{\sf out} \geq t_{\sf opt}$. As $t_{\sf out} \geq 2$, if $t_{\sf opt} =2$, then the claim trivially holds. Thus, we assume that $t_{\sf opt} \geq 3$. Let $\mcal{W}=(W_1,W_2,$ $\cdots, W_{t_{\sf opt}})$ be a $P_{t_{\sf opt}}$-witness structure of $G$, such that there is $i\in [{t_{\sf opt}}]$ with $|N[\bigcup_{j \in [i]}W_j]| \leq \alpha n$ and $|N[\bigcup_{j \in [{t_{\sf opt}}] \setminus [i]}W_j]| \leq \alpha n$. Let $S= \bigcup_{j \in [i]}W_j$ and $\overline S= \bigcup_{j \in [{t_{\sf opt}}] \setminus [i]}W_j$. Note that $\overline S= V(G) \setminus S$ and $S,\overline S\in \calS$. Let $\calW_1 = (W_1,W_2,\cdots, W_{i})$ and $\calW_2 = (W_{i+1},W_{i+2},\cdots, W_{t_{\sf opt}})$. Note that $\calW_1$ is a $P_i$-witness structure of $G[S]$, such that $\Phi(S) \subseteq W_i$. Similarly, $\calW_2$ is a $P_{{t_{\sf opt}}-i}$-witness structure of $G[\overline S]$, such that $\Phi(\overline S) \subseteq W_{i+1}$. From Lemma~\ref{lem:correct-EPPC}, we know that the algorithm has correctly computed the values $\Gamma[S]$ and $\Gamma[\overline S]$. From the above discussions we can conclude that $\Gamma[S] \geq i$ and $\Gamma[\overline S] \geq {t_{\sf opt}}-i$. Thus, we can conclude that $t_{\sf out} \geq \Gamma[S]+ \Gamma[\overline S] \geq {t_{\sf opt}}$. 

Next, we show that $t_{\sf out} \leq t_{\sf opt}$. As $t_{\sf opt} \geq 2$, if $t_{\sf out}=2$, the condition $t_{\sf out} \leq t_{\sf opt}$ is trivially satisfied. Now we consider the case when $t_{\sf out}=3$. In this case, there is a set $S\in \calS$, such that $V(G)\setminus S\in \calS$ and $t=\Gamma[S] + \Gamma[V(G)\setminus S]$. From Lemma~\ref{lem:correct-EPPC}, the algorithm has correctly computed $q_1=\Gamma[S]$ and $q_2=\Gamma[V(G)\setminus S]$. Thus, there is a $P_{q_1}$-witness structure $\calW_1=(W_1,W_2,\cdots, W_{q_1})$ for $G[S]$, such that $\Phi(S) \subseteq W_{q_1}$. Similarly, there is a $P_{q_2}$-witness structure $\calW_2=(W'_1,W'_2,\cdots, W'_{q_2})$ for $G[V(G) \setminus S]$, such that $\Phi(V(G) \setminus S) \subseteq W'_{q_2}$. Recall that $t=q_1+q_2$. We will show that $\calW= (W_1,W_2,\cdots, W_{q_1}, W'_{q_2}, \cdots, W'_2, W'_1)$ is a $P_t$-witness structure of $G$, such that there is $i\in [t]$ with $|N[\bigcup_{j \in [i]}W_j]| \leq \alpha n$ and $|N[\bigcup_{j \in [t] \setminus [i]}W_j]| \leq \alpha n$. Lemma~\ref{lem:correct-EPPC} and connectedness of $G$ implies that $\calW$ is a $P_t$-witness structure of $G$. Moreover, as $S,V(G)\setminus S \in \calS$, for $i=q_1$, we have $|N[\bigcup_{j \in [i]}W_j]| \leq \alpha n$ and $|N[\bigcup_{j \in [t] \setminus [i]}W_j]| \leq \alpha n$. Hence, we can conclude that $t_{\sf out} \leq t_{\sf opt}$. 
\end{proof}

\begin{lemma}\label{lem:correct-bpc-run}
The algorithm presented for \BPC\ runs in time $\calO^\star(2^{\alpha n})$.  
\end{lemma}
\begin{proof} The algorithm for \BPC\ calls Algorithm~\ref{method:partial-EPPC} on input $(G, \alpha)$ and iterates over all the values in $\Gamma$. Hence, Lemma~\ref{lem:correct-runtime-EPPC} implies that the running time of algorithm presented for \BPC\ is $\calO^\star(2^{\alpha n})$.
\end{proof}

%% file: disjoint-conn.tex
\subsection{Algorithm for \TDCPC}\label{sub-sec:dis-conn}
We formally define the problem \TDCPC\ in the following (also see Figure~\ref{fig:intuition-heavy-adj}). 

\begin{figure}[t]
  \centering
  \includegraphics[scale=0.75]{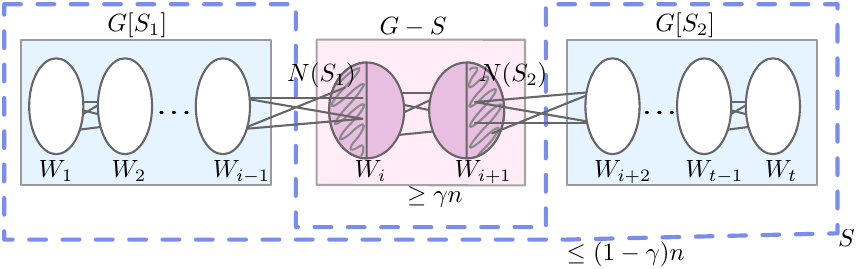}
  \caption{An intuitive illustration of the algorithm for \TDCPC.}\label{fig:intuition-heavy-adj}
\end{figure}

\defproblemout{\TDCPC}{A graph $G$ on $n$ vertices and a fraction $0<\gamma \leq 1$.}{Largest integer $t\geq 3$ for which $G$ has a $P_t$-witness structure $\mcal{W}=(W_1,W_2,$ $\cdots, W_t)$, such that there is $i\in [t-1]$ for which the following conditions hold: 1) $|W_i \cup W_{i+1}| \geq \gamma n$ and 2) $|N[\bigcup_{j\in [i-1]} W_j]|, |N[\bigcup_{j\in [t]\setminus [i+1]} W_j]| \leq (1-\gamma/2) n$. Moreover, if no such $t$ exists, then output $2$.
}
\vspace{0.2cm}

We design an algorithm for \TDCPC\ running in time $\mcal{O}^\star(2^{(1-\gamma/2)n} + c^n)$, where $c=\max_{\gamma \leq \delta \leq 1}\{1.7804^\delta\cdot g(1-\delta)\}$. The first term in the running time expression will be due to a call made to \EPPC\ with $\rho=(1-\gamma/2)$, and the second term will be due to enumerating sets of size at most $(1-\gamma) n$ and running the algorithm for solving \textsc{2-Disjoint Connected Subgraphs} for an instance created for each of them, using the algorithm of Telle and Villanger~\cite{telle2013connecting}. 

\iffalse
\begin{algorithm}[t]
  \KwIn{A graph $G$ and a fraction $0<\gamma \leq 1$.}
  \KwOut{An integer $t$.}
  Initialize $t = 2$\;
  
Let $\calS =\{S\subseteq V(G) \mid |S| \leq (1 - \gamma) n \mbox{ and } G[S] \mbox{ has exactly two connected}$ $\mbox{components } G[S_1],G[S_2], \mbox{ s.t. } |N[S_1]|, |N[S_2]| \leq (1-\gamma/2)n \}$\;

Let $\what \calS =\{\what S\subseteq V(G) \mid |N[\what S]| \leq (1 - \gamma/2) n \mbox{ and } G[\what S] \mbox{ is connected}\}$\; 

Compute the value of $\Gamma[S]$, for each $S\in \what \calS$, by computing the table ${\Gamma}$ = \EPPC$(G, 1 - \gamma/2)$\; %Run Algorithm~\ref{method:dynamic-programming} on input $(G, 1 - \gamma/2)$ to compute table $\calT$

  \For{$S \in \calS$}{
  Let $S_1$ and $S_2$ be the two connected components of $G[S]$\;
%    \tcc{\textbf{Continue} if $G[S]$ do not have two connected components}
%    \tcc{\textbf{Continue} if $|N[S_1]|$ or $|N[S_2]|$ are not at most $(1 - \gamma/2)n$}
    \If{$(G-S, N_G(S_1), N_G(S_2))$ is a \yes\ instance of \textsc{$2$-DCS}}{
        $t = \max\{t, \Gamma[S_1] + \Gamma[S_2] + 2\}$
    }
  }
  \Return $t$\;
  \caption{Algorithm for \TDCPC}
  \label{method:disjoint-conn-2}
\end{algorithm}
\fi

\begin{algorithm}
    \caption{Algorithm for \TDCPC}
	\label{method:disjoint-conn-2}
	\begin{algorithmic}[1]
	\REQUIRE { A graph $G$ and a fraction $0<\gamma \leq 1$.}
	\ENSURE { An integer $t$.}
		%\STATE{\algorithmicrequire{ A graph $G$ and a fraction $0<\gamma \leq 1$.}}
		%\STATE{\algorithmicensure{ An integer $t$.}}
		\vspace{0.2cm}
		\STATE{Initialize $t = 2$}
		\STATE{Let $\calS =\{S\subseteq V(G) \mid |S| \leq (1 - \gamma) n \mbox{ and } G[S] \mbox{ has exactly two connected}$ $\mbox{components } G[S_1],G[S_2], \mbox{ s.t. } |N[S_1]|, |N[S_2]| \leq (1-\gamma/2)n \}$ \label{step:compute-S}}
		\STATE{Let $\what \calS =\{\what S\subseteq V(G) \mid |N[\what S]| \leq (1 - \gamma/2) n \mbox{ and } G[\what S] \mbox{ is connected}\}$. Compute the value of $\Gamma[S]$, for each $S\in \what \calS$, by computing the table ${\Gamma}$ = \EPPC$(G, 1 - \gamma/2)$ \label{step:compute-hat-S}}
		\FOR{$S \in \calS$ \label{step:start-for1}}
		\STATE{  Let $S_1$ and $S_2$ be the two connected components of $G[S]$}
		\IF{$(G-S, N_G(S_1), N_G(S_2))$ is a \yes\ instance of \textsc{$2$-DCS}}
		\STATE{ $t = \max\{t, \Gamma[S_1] + \Gamma[S_2] + 2\}$}
		\ENDIF
		\ENDFOR
		\RETURN{$t$}
	\end{algorithmic}
\end{algorithm}

We now formally describe our algorithm. The algorithm will output an integer $t$, which is initially set to $2$. Let $\calS =\{S\subseteq V(G) \mid |S| \leq (1 - \gamma) n \text{ and } G[S] \allowbreak \text{ has exactly two connected components } G[S_1],G[S_2], \text{ such that } |N[S_1]|, |N[S_2]| \leq \\ (1 - \gamma/2)n \}$. Let $\what \calS =\{\what S\subseteq V(G) \mid |N[\what S]| \leq (1 - \gamma/2) n \mbox{ and } G[\what S] \mbox{ is connected}\}$. The algorithm will now compute a table $\Gamma$, which has an entry $\Gamma[\what S]$, for each $\what S\in \what \calS$. The definition of $\Gamma$ is the same as that in Section~\ref{sub-sec:dp}, where $\rho=1-\gamma/2$. That is, for $\what S\in \what \calS$, $\Gamma[\what S]$ is the largest integer $q \geq 1$ for which $G[\what S]$ can be contracted to $P_q$ with a $P_q$-witness structure $\mcal{W}=(W_1,W_2,$ $\cdots, W_q)$ of $G[\what S]$, such that $\Phi(\what S) \subseteq W_q$. Compute the value of $\Gamma[\what S]$, for each $\what S\in \what \calS$, by using \EPPC$(G, 1 - \gamma/2)$. For each $S\in \calS$, the algorithm does the following. Recall that $G[S]$ has exactly two connected components. Let the two connected components in $G[S]$ be $G[S_1]$ and $G[S_2]$, where $S_1\cup S_2=S$. Recall that $|N[S_1]|, |N[S_2]| \leq (1-\gamma/2)n$. Thus, $S_1,S_2\in \what\calS$. If $(G-S, N_G(S_1), N_G(S_2))$ is a yes-instance of \textsc{$2$-DCS}, then the algorithm sets $t = \max\{t, \Gamma[S_1] + \Gamma[S_2] + 2\}$, and otherwise, it moves to the next set in $\what \calS$. Finally, the algorithm outputs $t$. This completes the description of the algorithm. See Algorithm~\ref{method:disjoint-conn-2}.

In the following two lemmas we present the correctness and runtime analysis of the algorithm, respectively.

\begin{lemma}\label{lemma:correct-TDCPC}
The algorithm presented for \TDCPC\ is correct. 
\end{lemma}
\begin{proof}
For an instance $(G,\gamma)$, suppose that $t_{\sf opt}$ is the solution to \TDCPC\ and $t_{\sf out}$ is the output returned by the algorithm. We will show that $t_{\sf out} = t_{\sf opt}$. 

Firstly, we show that $t_{\sf out} \geq t_{\sf opt}$. As $t_{\sf out} \geq 2$, if $t_{\sf opt} =2$, then the claim trivially holds. Thus, we assume that $t_{\sf opt} \geq 3$. Let $\mcal{W}=(W_1,W_2,$ $\cdots, W_{t_{\sf opt}})$ be a $P_{t_{\sf opt}}$-witness structure of $G$, such that there is $i\in [t_{\sf opt}-1]$ for which the following conditions hold: 1) $|W_i \cup W_{i+1}| \geq \gamma n$ and 2) $|N[\bigcup_{j\in [i-1]} W_j]|, |N[\bigcup_{j\in [t_{\sf opt}]\setminus [i+1]} W_j]| \leq (1 - \gamma/2) n$. Let $Z= W_{i} \cup W_{i+1}$ and $S=V(G) \setminus Z$. As $|Z| \geq \gamma n$, we have $|S| \leq (1-\gamma)n$. Also, as $\calW$ is a $P_{t_{\sf opt}}$-witness structure of $G$, $G[S]$ has exactly two connected components $G[S_1]$ and $G[S_2]$, where $S_1= \bigcup_{j\in [i-1]} W_j$ and $S_2= \bigcup_{j\in [t_{\sf opt}]\setminus[i+1]}W_j$. Note that we have $|N[S_1]| \leq (1-\gamma/2)n$ and $|N[S_2]| \leq (1-\gamma/2)n$. From the above discussions, we can conclude that $S\in \calS$ and $S_1,S_2\in \what \calS$. By Lemma~\ref{lem:correct-EPPC}, the values of $\Gamma[S_1]$ and $\Gamma[S_2]$ are computed correctly, we can conclude that $\Gamma[S_1]+ \Gamma[S_2] \geq t_{\sf opt} -2$. Also, $(W_{i}, W_{i+1})$ is a solution to the instance $(G-S, N_G(S_1), N_G(S_2))$, and hence it is a yes-instance of {\sc $2$-DCS}. Thus,  $t_{\sf out} = t \geq \Gamma[S_1]+ \Gamma[S_2] +2 \geq t_{\sf opt}$. 

Next, we show that $t_{\sf out} \leq t_{\sf opt}$. As $t_{\sf opt} \geq 2$, if $t_{\sf out}=2$, the condition $t_{\sf out} \leq t_{\sf opt}$ is trivially satisfied. Now we consider the case when $t_{\sf out}\geq 3$. In this case, there is a set $S\in \calS$, such that $(G-S, N_G(S_1), N_G(S_2))$ is a yes-instance of {\sc $2$-DCS} and $t_{\sf out} = \Gamma[S_1] + \Gamma[S_2] +2$, where $G[S_1]$ and $G[S_2]$ are the two connected component of $G[S]$. From Lemma~\ref{lem:correct-EPPC} it follows that the algorithm has correctly computed $q_1=\Gamma[S_1]$ and $q_2=\Gamma[S_2]$. Thus, there is a $P_{q_1}$-witness structure $\calW_1=(W_1,W_2,\cdots, W_{q_1})$ for $G[S_1]$, such that $N(S_1) \subseteq W_{q_1}$. Similarly, there is a $P_{q_2}$-witness structure $\calW_2=(W'_1,W'_2,\cdots, W'_{q_2})$ for $G[S_2]$, such that $N(S_2) \subseteq W'_{q_2}$. \sloppy Note that $q_1+q_2 \leq t_{\sf opt}-2$. Let $(Z_1,Z_2)$ be a solution to {\sc $2$-DCS} in $(G-S, N_G(S_1), N_G(S_2))$. Note that $\calW= (W_1,W_2,\cdots, W_{q_1}, Z_1,Z_2,W'_{q_2}, \cdots, W'_2, W'_1)$ $= (W_1,W_2,\cdots, W_{q_1}, W_{q_1+1},W_{q_1+2},\cdots, W_{q_1+q_2+2})$ is a $P_{t_{\sf out}}$-witness structure of $G$, such that: 1) $|W_{q_1+1} \cup W_{q_1+2}| \geq \gamma n$ and 2) $|N[\bigcup_{j\in [q_1]} W_j]|, |N[\bigcup_{j\in [t_{\sf out}]\setminus [q_1+2]} W_j]| \leq (1-\gamma/2) n$. Thus, we can conclude that $t_{\sf out} \leq t_{\sf opt}$.
\end{proof}

\begin{lemma}\label{lem:correct-runtime-TDCPC}
\sloppy The algorithm presented for \TDCPC\ runs in time $\mcal{O}^\star(2^{(1-\gamma/2)n} + c^n)$, where $c=\max_{\gamma \leq \delta \leq 1}\{1.7804^\delta\cdot g(1-\delta)\}$.  
\end{lemma}
\begin{proof}
Using Observation~\ref{obs:subset-no}, Step~\ref{step:compute-S} of the algorithm can be executed in time $\calO^\star((g(1-\gamma))^n)$, which is bounded by $\calO^\star(c^n)$. 
%Step~\ref{step:compute-hat-S} of the algorithm can be executed in time $\mcal{O}^\star(2^{(1-\gamma/2)n})$, using Lemma~\ref{lemma:no-conn-comps}. 
Step~\ref{step:compute-hat-S} of the algorithm can be executed in time $\mcal{O}^\star(2^{(1-\gamma/2)n})$, by Lemma~\ref{lem:correct-EPPC} and Lemma~\ref{lem:correct-runtime-EPPC}. We now argue about the time required for the \emph{for}-loop starting at Step~\ref{step:start-for1} (all the remaining steps can be executed in constant time). The number of sets in $\calS$ of size at most $(1-\delta)n$ is bounded by $(g(1-\delta))^n$. For each $\gamma \leq \delta \leq 1$, and each set $S \in \calS$ of size at most $(1-\delta)n$, we resolve the instance $(G-S, N_G(S_1), N_G(S_2))$ of {\sc $2$-DCS}, where $G[S_1]$ and $G[S_2]$ are the two connected components of $G-S$. Note that the number of vertices in $G-S$ is bounded by $\delta n$, and hence using Proposition~\ref{prop:exact-2-con}, we can resolve the instance $(G-S, N_G(S_1), N_G(S_2))$ of {\sc $2$-DCS} in time $\calO^\star(1.7804^{\delta n})$. From the above discussions we can conclude that the running time of the algorithm is bounded by $\mcal{O}^\star(2^{(1-\gamma/2)n} + c^n)$, where $c=\max_{\gamma \leq \delta \leq 1}\{1.7804^\delta\cdot g(1-\delta)\}$.
\end{proof}

%% file: enumerating-subsets.tex
\subsection{Algorithm for \SOEPC}
\label{sub-sec:enum}
We formally define the problem \SOEPC\ in the following. 

\defproblemout{\SOEPC}{A graph $G$ on $n$ vertices and a fraction $0<\beta\leq 1$.}{Largest integer $t$ for which $G$ can be contracted to $P_t$, with $\mcal{W}=(W_1,W_2,\cdots, W_t)$ as a $P_t$-witness structure of $G$, such that $|\os_{\mcal{W}}| \leq \beta n/2$ or $|\es_{\mcal{W}}| \leq \beta n/2$, where $\os_{\mcal{W}}=\bigcup_{i\in [\lceil{t/2}\rceil]} W_{2i-1}$ and $\es_{\mcal{W}}=\bigcup_{i\in [\lfloor{t/2}\rfloor]} W_{2i}$.
}
\vspace{0.2cm}

In this section, we design an algorithm for \SOEPC\ running in time $\mcal{O}^\star(c^n)$, where $c=g(\beta/2)$. 

Let $(G,\beta)$ be an instance of \SOEPC. The algorithm is fairly simple. It starts by enumerating all ``potential candidates'' for $\os_{\mcal{W}}$ (resp. $\es_{\mcal{W}}$), i.e., the set of all subsets of $V(G)$ of size at most $\beta n/2$. Then, for each such ``potential set'', it contracts $G$ appropriately, and finds the length of the path to which $G$ is contracted (and stores $0$, if the contracted graph is not a path). Finally, it returns the maximum over such path lengths.

We now move to formal description of the algorithm. We start by enumerating the set of all subsets of $V(G)$ of size at most $\beta n/2$. That is, $\mcal{S} = \{S\subseteq V(G) \mid |S| \leq \beta n/2\}$. Note that $\mcal{S}$ can be computed in time $\calO^{\star}(g(\beta/2)^n)$, using Observation~\ref{obs:subset-no}. For each $S\in \mcal{S}$ the algorithm does the following. Let $\mcal{C}_S$ and $\overline{\mcal{C}}_S$ be the set of connected components of $G[S]$ and $G-S$, respectively. Let $G_S$ be the graph obtained from $G$ by contracting each $C\in \mcal{C}_S \cup \overline{\mcal{C}}_S$ to a single vertex. Set ${\sf len}_S= |V(G_S)|$, if $G_S$ is a path, and ${\sf len}_S= 0$, otherwise. Finally, the algorithm returns $\max_{S\in \mcal{S}}{\sf len}_S$. 

 In the following lemma we prove the correctness and runtime analysis of the algorithm.

\begin{lemma}\label{lemma:alg-SOEPC}
The algorithm presented for \SOEPC\ is correct and runs in time $\calO^{\star}(g(\beta/2)^n)$. 
\end{lemma}
\begin{proof}
Clearly, the algorithm presented for \SOEPC\ runs in time $\calO^{\star}(g(\beta/2)^n)$. Now we prove the correctness of the algorithm. 

In the forward direction, assume that $G$ is contractible to $P_{t_{opt}}$, where $(W_1,W_2,\cdots,$ $W_{t_{opt}})$ is a $P_{t_{opt}}$-witness structure of $G$, such that $|\os_{\mcal{W}}| \leq \beta n/2$ or $|\es_{\mcal{W}}| \leq \beta n/2$, where $\os_{\mcal{W}}=\bigcup_{i\in [\lceil{t_{opt}}\rceil]} W_{2i-1}$ and $\es_{\mcal{W}}=\bigcup_{i\in [\lfloor{t_{opt}}\rfloor]} W_{2i}$. We will show that the algorithm outputs $t_{out} \geq t_{opt}$. We assume that $|\os_{\mcal{W}}| \leq \beta n/2$. (The case when $|\es_{\mcal{W}}| \leq \beta n/2$ can be argued analogously.) Let $S=\os_{\mcal{W}}$. Note that $S\in \mcal{S}$. The set of connected components in $G[S]$ is precisely $\mcal{C}_S = \{G[W_{2i-1}] \mid i \in [\lceil{t_{out}/2}\rceil]\}$. Also, the set of connected components in $G-S$ is precisely $\overline{\mcal{C}}_S = \{G[W_{2i + 1}] \mid i \in [\lfloor{t_{out}/2}\rfloor]\}$. Thus, $G_S$ is isomorphic to $P_{t_{opt}}$. Thus, the output of the algorithm $t_{out} = \max_{S'\in \mcal{S}}{\sf len}_{S'} \geq t_{opt}$, as $S\in \mcal{S}$ and ${\sf len}_{S} = t_{opt}$. 

For the other direction, let $t_{out}$ be the output of the algorithm. Note that $t_{out} \geq 1$, as $\emptyset \in \mcal{S}$ and $G_\emptyset$ is a single vertex (as $G$ is connected). Consider $S\in \mcal{S}$, such that ${\sf len}_S = t_{out}$. Note that a $P_t$-witness set for $G$ is $\mcal{W} = \{V(C) \mid C\in \mcal{C}_S\} \cup \{V(C) \mid C\in \overline{\mcal{C}}_S\}$. Thus one of $\os_{\mcal{W}} =S$ or $\es_{\mcal{W}} =S$ must hold. Moreover, as $S\in \mcal{S}$, we have $|S| \leq \beta n/2$. This concludes the proof. 
%Note that $t\geq 2$, as Observation~\ref{obs:contract-P2} and~\ref{obs:singleton-end-bags}
\end{proof}

%% file: disjoint-conn-3.tex
\subsection{Algorithm for \NSOEPC}
\label{sub-sec:dis-conn-3}
%In this section, we consider the problem \NSOEPC\ (to be defined shortly). We note that whenever we are looking for a $P_t$-witness structures for a graph, then the first and the last witness sets will contain a single vertex.  

We formally define the problem \NSOEPC\ in the following (also see Figure~\ref{fig:intuition-near-odd-even}). 

\defproblemout{\NSOEPC}{A graph $G$ on $n$ vertices and a fraction $0<\epsilon \leq 1$.}{Largest integer $t \geq 3$ for which there is a $P_t$-witness structure $\mcal{W}=(W_1,W_2,\cdots, W_t)$ of $G$, for
 which there is $i \in \{2,3,\cdots, t-1\}$, such that if $i$ is odd, then $|\os_{\mcal{W}} \setminus W_i| \leq \epsilon n$ and otherwise, $|\es_{\mcal{W}}\setminus W_i| \leq \epsilon n$. Here, $\os_{\mcal{W}}=\bigcup_{i\in [\lceil{t/2}\rceil]} W_{2i-1}$ and $\es_{\mcal{W}}=\bigcup_{i\in [\lfloor{t/2}\rfloor]} W_{2i}$. If no such $t\geq 3$ exists, then output $2$. 
}
\vspace{0.2cm}

\begin{figure}[t]
  \centering
  \includegraphics[scale=0.70]{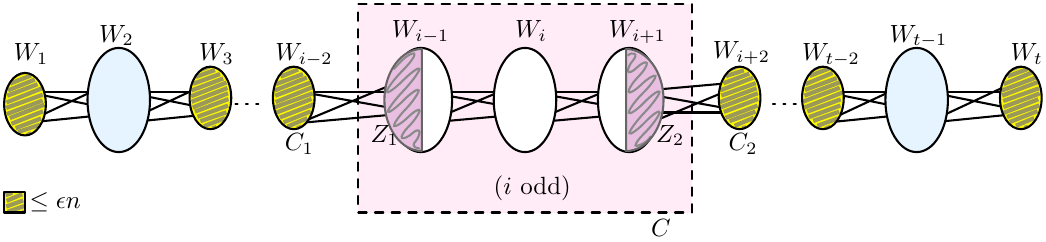}
  \caption{An intuitive illustration of the algorithm for \NSOEPC.}\label{fig:intuition-near-odd-even}
\end{figure}

We design an algorithm for \NSOEPC\ running in time $\mathcal{O}^{\star}(c^{n})$ where $c = \max_{0 \le \delta \le \epsilon} \{1.88^{(1 - \delta)}~\cdot~g(\delta)\}$. The second term in multiplicative factor will be due to enumeration of sets, and the first term will be due to calls made to the algorithm for \textsc{$3$-Disjoint Connected Subgraphs}, from Section~\ref{sec:prob-dis-conn-3}. 

Let $(G,\epsilon)$ be an instance of \NSOEPC. We start by explaining the intuitive idea behind our algorithm (see Figure~\ref{fig:intuition-near-odd-even}). Consider a $P_t$-witness structure $\mcal{W}=(W_1,W_2,$ $\cdots, W_t)$ of $G$, for
 which there is $i \in \{2,3,\cdots, t-2\}$, such that if $i$ is odd, then $|\os_{\mcal{W}} \setminus W_i| \leq \epsilon n$ and otherwise, $|\es_{\mcal{W}}\setminus W_i| \leq \epsilon n$. In the above, $\os_{\mcal{W}}=\bigcup_{i\in [\lceil{t/2}\rceil]} W_{2i-1}$ and $\es_{\mcal{W}}=\bigcup_{i\in [\lfloor{t/2}\rfloor]} W_{2i}$. Let us consider the case when $i$ is odd (the other case is symmetric). Let $S = \os_{\mcal{W}} \setminus W_i$. (The union of vertices from yellow sets in Figure~\ref{fig:intuition-near-odd-even} is the set $S$.) As $|S| \leq \epsilon n$, the algorithm starts by enumerating all ``potential candidates'' for the set $S$. All the components of $G-S$, except for the component $C$, containing $W_i$, must each  be contracted to a single vertex. Similarly, each connected components of $G[S]$ must be contracted to a single vertex. Moreover, the component containing $W_i$ must be ``split'' into three sets. The first and the last sets in the ``split'' must contain the neighbors of $W_{i-2}$ and $W_{i+2}$ in $C$, respectively. To obtain such a ``split'', we use the algorithm for \textsc{$3$-Disjoint Connected Subgraphs} that we designed in Section~\ref{sec:prob-dis-conn-3}.
 
We now formally describe our algorithm. (A pseudo code of our algorithm is presented in Algorithm~\ref{method:disjoint-conn-3}.) The algorithm will output an integer $t$, which is initially set to $2$. Let $\calS =\{S\subset V(G) \mid |S| \leq \epsilon n\}$. For each $S \in \calS$, the algorithm does the following.  Let $\mcal{C}_S$ and $\overline{\mcal{C}}_S$ 
  be the sets of connected components in $G[S]$ and $G-S$, respectively. Let $H_S$ be obtained from $G$ by contracting component in $\mcal{C}_S \cup \overline{\mcal{C}}_S$ to single vertices. That is, $H_S$ has a vertex $v_{C}$ for each $C \in \mcal{C}_S \cup \overline{\mcal{C}}_S$, and two vertices $v_{C}, v_{C'} \in V(H_S)$ are adjacent in $H_S$ if and only if $C$ and $C'$ are adjacent in $G$. If $H_S$ is not a path, then the algorithm moves to the next set in $\calS$. Otherwise, for each $C^* \in \overline{\mcal{C}}_S$ it does the following. Intuitively speaking, $C^*$ is the current guess for the component containing vertices from $W_i$ for the witness structure that we are looking for. Note that $C^*$ can be adjacent to at most two components from ${\mcal{C}}_S$, as $H_S$ is a path. Moreover, $C^*$ must be adjacent to at least one component from ${\mcal{C}}_S$, as $G$ is connected and $S$ is a strict subset of $V(G)$, i.e., $S\neq V(G)$. Let $C_1$ be a component from ${\mcal{C}}_S$ that is adjacent to $C^*$ in $G$, and $Z_1 = N(C_1) \cap V(C^*)$. Let $C_2 \in {\mcal{C}}_S \setminus \{C_1\}$ be a component of $G[S]$ that is adjacent to $C^*$, and $Z_2=N(C_2) \cap V(C^*)$. If such a $C_2$ does not exist, then we set $Z_2=\emptyset$. If $(G[C^*], Z_1, Z_2)$ is a yes-instance of \textsc{$3$-DCS}, then the algorithm updates $t = \max\{t, |V(H_S)| + 2\}$. After finishing the processing for each $S\in \mcal{S}$, the algorithm outputs $t$. This finishes the description of our algorithm. 
   
\iffalse  
\begin{algorithm}[t]
  \KwIn{A graph $G$ and a fraction $0<\epsilon \leq 1$.}
  \KwOut{An integer $t$.}
  Initialize $t = 2$\;  
  Let $\calS =\{S\subset V(G) \mid |S| \leq \epsilon n\}$\;
  \For{$S \in \calS$ }{
  Let $\mcal{C}_S$ and $\overline{\mcal{C}}_S$ 
  be the sets of connected components in $G[S]$ and $G-S$, respectively\;
  Let $H_S$ be obtained from $G$ by contracting component in $\mcal{C}_S \cup \overline{\mcal{C}}_S$ to single vertices\;
    \If{$H_S$ is a path}{
    \For{$C^*\in \overline{\mcal{C}}_S$}{
     Let $C_1 \in {\mcal{C}}_S$ be a component of $G[S]$ that is adjacent to $C^*$, and $Z_1=N(C_1) \cap V(C^*)$\;
      Let $C_2 \in {\mcal{C}}_S \setminus \{C_1\}$ be a component of $G[S]$ that is adjacent to $C^*$, and $Z_2=N(C_2) \cap V(C^*)$, is such a $C_2$ does not exist, then $Z_2=\emptyset$\;
      \If{$(G[C^*], Z_1, Z_2)$ is a yes-instance of \textsc{$3$-DCS}}{
        $t = \max\{t, |V(H_S)| + 2\}$\;
      }
    }
  }
}  
  \Return $t$\;
  \caption{Algorithm for \NSOEPC.}
  \label{method:disjoint-conn-3}
\end{algorithm} 
\fi

\begin{algorithm}[t]
   \caption{Algorithm for \NSOEPC.}
	\label{method:disjoint-conn-3}
	\begin{algorithmic}[1]
	\REQUIRE { A graph $G$ and a fraction $0<\epsilon \leq 1$.}
	\ENSURE { An integer $t$.}
%	  \STATE{\algorithmicrequire{ A graph $G$ and a fraction $0<\epsilon \leq 1$.}}
%	  \STATE{\algorithmicensure{ An integer $t$.}}
	  \vspace{0.2cm}
	  \STATE{Initialize $t = 2$}
	  \STATE{Let $\calS =\{S\subset V(G) \mid |S| \leq \epsilon n\}$ \label{step:compute-S-2}}
       \FOR{$S \in \calS$ \label{step:for-loop-S}}
			\STATE{Let $\mcal{C}_S$ and $\overline{\mcal{C}}_S$ be the sets of connected components in $G[S]$ and $G-S$, resp.}
		
		\STATE{Let $H_S$ be obtained from $G$ by contracting components in $\mcal{C}_S \cup \overline{\mcal{C}}_S$ to single vertices}
		
		\IF{$H_S$ is a path}
		\FOR{$C^*\in \overline{\mcal{C}}_S$ \label{step:for-loop-c}}
		\STATE{Let $C_1 \in {\mcal{C}}_S$ be a component of $G[S]$ that is adjacent to $C^*$, and $Z_1=N(C_1) \cap V(C^*)$}
		\STATE{Let $C_2 \in {\mcal{C}}_S \setminus \{C_1\}$ be a component of $G[S]$ that is adjacent to $C^*$, and $Z_2=N(C_2) \cap V(C^*)$, if such a $C_2$ does not exist, then $Z_2=\emptyset$}
		\IF{$(G[C^*], Z_1, Z_2)$ is a yes-instance of \textsc{$3$-DCS} \label{step:check-3DCS}}
		\STATE{$t = \max\{t, |V(H_S)| + 2\}$}
		\ENDIF
		\ENDFOR
		\ENDIF
		\ENDFOR
  \RETURN{$t$}
		
	\end{algorithmic}
\end{algorithm}

In the following two lemmas we present the correctness and runtime analysis of the algorithm, respectively. 

\begin{lemma}\label{lemma:correct-NSOEPC}
The algorithm presented for \NSOEPC\ is correct. 
\end{lemma}
\begin{proof}
For an instance $(G,\epsilon)$, suppose that $t_{\sf opt}$ is the solution to \NSOEPC\ and $t_{\sf out}$ is the output returned by the algorithm. We will show that $t_{\sf out} = t_{\sf opt}$. 

Firstly, we show that $t_{\sf out} \geq t_{\sf opt}$. As $t_{\sf out} \geq 2$, if $t_{\sf opt} =2$, then the claim trivially holds. Thus, we assume that $t_{\sf opt} \geq 3$. Let $\mcal{W}=(W_1,W_2,$ $\cdots, W_{t_{\sf opt}})$ be a $P_{t_{\sf opt}}$-witness structure of $G$, for which there is $i \in \{2,3,\cdots, {t_{\sf opt}}-1\}$, such that if $i$ is odd, then $|\os_{\mcal{W}} \setminus W_i| \leq \epsilon n$ and otherwise, $|\es_{\mcal{W}}\setminus W_i| \leq \epsilon n$. In the above, $\os_{\mcal{W}}=\bigcup_{i\in [\lceil{t_{\sf opt}/2}\rceil]} W_{2i-1}$ and $\es_{\mcal{W}}=\bigcup_{i\in [\lfloor{t_{\sf opt}/2}\rfloor]} W_{2i}$. We consider the case when $i$ is odd. (The case when $i$ is even can be argued analogously.) Let $S= \os_{\mcal{W}} \setminus W_i$. As $|S| \leq \epsilon n$, we have $S\in \calS$. Note that $H_S$ is a path. Let $C^* \in \overline{\mcal{C}}_S$ be the connected component of $G-S$ containing $W_i$. Let $C_1 \in \mcal{C}_S$ be a connected component of $G[S]$ adjacent to $C^*$, and $Z_1 = N(C_1) \cap V(C^*)$. Consider $C_2\in \mcal{C}_S\setminus \{C_1\}$ that is adjacent to $C^*$, and let $Z_2 = N(C_2) \cap V(C^*)$. If such a $C_2$ does not exist, then set $Z_2=\emptyset$. Note that $(G[C^*], Z_1, Z_2)$ is a yes-instance of \textsc{$3$-DCS}, as $(W_{i-1},W_i,W_{i+1})$ is a solution to it. In the above we rely on the fact that $i\in \{2,3,\cdots, t-1\}$, and thus each of $W_{i-1}$ and $W_{i+1}$ are non-empty. From the above discussions we can conclude that $t_{\sf out} \geq t_{\sf opt} = V(H_S) +2 $ (as $C^*$ is split into three witness sets). 

Next, we show that $t_{\sf out} \leq t_{\sf opt}$. As $t_{\sf opt} \geq 2$, if $t_{\sf out}=2$, the condition $t_{\sf out} \leq t_{\sf opt}$ is trivially satisfied. Now we consider the case when $t_{\sf out}\geq 3$. There is $S\in \calS$ for which $H_S$ is a path and there is $C^* \in \overline{C}_S$, for which the instance $(G[C^*],Z_1,Z_2)$ is a yes-instance of \textsc{$3$-DCS}. Let $(V^*_1,U^*,V^*_2)$ be a solution to \textsc{$3$-DCS} for the instance $(G[C^*],Z_1,Z_2)$. Let $\calW' = \mcal{C}_S \cup (\overline{\mcal{C}}_S \setminus \{C^*\}) \cup \{V^*_1,U^*,V^*_2\}$. Note that $|\calW'| = |V(H_S)| + 2$ and $\calW'$ is a $P_{t_{\sf out}}$-witness structure for $G$. Let $\calW=(W_1,W_2,\cdots, W_{t_{\sf out}})$ be the ordered witness structure corresponding to the $P_{t_{\sf out}}$-witness structure $\calW'$ of $G$. Note that there is $i\in \{2,3,\cdots,t_{\sf out}-1\}$, such that $V(C^*) \subseteq W_{i-1}\cup W_{i} \cup W_{i+1}$. Thus we can conclude that $\calW$ is a $P_{t_{\sf out}}$-witness structure of $G$, for which for which there is $i \in \{2,3,\cdots, {t_{\sf out}}-1\}$, such that if $i$ is odd, then $|\os_{\mcal{W}} \setminus W_i| \leq \epsilon n$ and otherwise, $|\es_{\mcal{W}}\setminus W_i| \leq \epsilon n$. From the above discussions we can conclude that $t_{\sf out} \leq t_{\sf opt}$. 
\end{proof}

\begin{lemma}\label{lem:correct-runtime-NSOEPC}
The algorithm presented for \NSOEPC\ runs in time $\mathcal{O}^{\star}(c^{n})$, where $c = \max_{0 \le \delta \le \epsilon} \{1.88^{(1 - \delta)}~\cdot~g(\delta)\}$.  
\end{lemma}
\begin{proof}
From Observation~\ref{obs:subset-no}, Step~\ref{step:compute-S-2} of the algorithm can be executed in time $\calO^{\star}([g(\epsilon)]^n)$. Also, $|\calS|$ is bounded by $\calO^{\star}([g(\epsilon)]^n)$. For a set $S\in \calS$, Step~\ref{step:check-3DCS} can be executed in time $\calO^{\star}(1.88^{n - |S|})$, from Theorem~\ref{thm:exact-3-con} and other steps can be executed in polynomial time. Hence, the running time algorithm can be bounded by $\calO^{\star}(c^{n})$, where $c = \max_{0 \le \delta \le \epsilon} \big\{1.88^{(1 - \delta)}~\cdot~g(\delta)\big\}$.
\end{proof}

%% file: algo.tex
\subsection{Algorithm for {\sc Path Contraction}}
\label{sub-sec:algo}
We are now ready to present our algorithm for \textsc{Path Contraction}. The algorithm calls four of the subroutines \SOEPC, \BPC, \TDCPC, and \ThDCPC\ for appropriate instances, and returns the maximum of their outputs. In the following theorem, we present the algorithm, which is the main result of this paper. 

%Following theorem is the main result of this paper.
 
\begin{theorem} 
\sloppy	\textsc{Path Contraction} admits an algorithm running in time $\mathcal{O}^{\star}(1.99987^n)$, where $n$ is the number of vertices in the input graph.  
\end{theorem}
\begin{proof}
We fix $\alpha, \beta, \gamma$ such that they satisfy following inequalities: $(1)\ 2 -\alpha - \beta/2 + \gamma/2  \le \alpha$; $(2)\ 1 - \gamma/2 \le \alpha$. These inequalities will be used in the later parts of the proof. We set $\alpha = 0.9996$, $\beta = 0.9885$, $\gamma = 0.9864$, and $\epsilon = 1 - \beta/2 - \gamma/2$. 
 
The algorithm for {\sc Path Contraction} is as follows. Let $G$ be the input graph. Let $t_1=$ \SOEPC$(G,\beta)$, $t_2=$\BPC$(G,\alpha)$, $t_3=$\TDCPC$(G,\gamma)$, and $t_4=$ \NSOEPC$(G,\epsilon)$. Furthermore, let $t^*=\max\{2,t_1,t_2,t_3,t_4\}$. The algorithm returns $t^*$. This finishes the description of the algorithm. 

%As the algorithm calls the Algorithm~\ref{method:enumerate}, Algorithm~\ref{method:dynamic-programming}, Algorithm~\ref{method:disjoint-conn-2}, Algorithm~\ref{method:disjoint-conn-3} with input $(G, \beta)$, $(G, \alpha)$, $(G, \gamma)$ and $(G, 1 - \beta/2 - \gamma/2)$, respectively. It returns the maximum among values obtained by these four algorithms. 

By Lemma~\ref{lemma:alg-SOEPC}, $t_1$ can be computed in time $\calO^{\star}(1.99987^n)$. By Lemma~\ref{lem:correct-bpc-run}, $t_2$ can be computed in time $\calO^{\star}(1.9995^n)$. From Lemma~\ref{lem:correct-runtime-TDCPC}, $t_3$ can be computed in time $\calO^{\star}(1.9133^n)$. From Lemma~\ref{lem:correct-runtime-NSOEPC}, $t_4$ can be computed in time $\calO^{\star}(1.9953^n)$. Thus, the running time of the algorithm is bounded by $\calO^{\star}(1.99987^n)$. 

We now prove the correctness of the algorithm. If the algorithm returns an integer $t$, then from Lemmas~\ref{lemma:alg-SOEPC},~\ref{lem:correct-bpc},~\ref{lemma:correct-TDCPC}, and~\ref{lemma:correct-NSOEPC}, it follows that $G$ is contractible to $P_t$. Now we prove the other direction. Suppose $G$ is contractible to $P_t$. We will show that $t^* \geq t$. 

%By Lemma~\ref{lemma:enumerate-correct}, Lemma~\ref{lemma:dynamic-programming-correct}, Lemma~\ref{lemma:disjoint-conn-2-correct} and Lemma~\ref{lemma:disjoint-conn-3-correct}  the running time for these algorithms for specified values of $\alpha, \beta, \gamma$ are $\calO^{\star}(1.99987^n)$, $\calO^{\star}(1.9994^n)$, $\calO^{\star}(1.8983^n)$ and $\calO^{\star}(1.9921^n)$ respectively. These lemmas also implies that if algorithm returns an integer $t$ then input graph can be contracted to $P_t$. To argue the correctness of main algorithm, it remains to argue that if $\ell$ is the largest integer such that $G$ can be contracted to $P_{\ell}$ than there exists a $P_{\ell}$-witness structure of $G$ which satisfies premises of either one of Lemma~\ref{lemma:enumerate-optimal}, Lemma~\ref{lemma:dynamic-programming-optimal}, Lemma~\ref{lemma:disjoint-conn-2-optimal}, or Lemma~\ref{lemma:disjoint-conn-3-optimal}.

Let $\calW=(W_1, W_2, \dots, W_{t})$ be a $P_{t}$-witness structure of $G$. We assume that $t \geq 3$, as otherwise, trivially, $t^*\geq t$ is satisfied. We also assume that $|W_1| = |W_t| = 1$ (see Observation~\ref{obs:singleton-end-bags}). Let $ \os = \bigcup_{i \in [\lceil t/2 \rceil]} W_{2i - 1}$ and $ \es = \bigcup_{i \in [\lfloor t/2 \rfloor]} W_{2i}$. For $i\in [t]$, we let $Q_i=\bigcup_{j\in [i]}W_j$ and $R_i=\bigcup_{j\in [t]\setminus [i-1]}W_j$. Note that $N(Q_i)$ (resp. $N(R_i)$) is contained in $W_{i+1}\subseteq Q_{i + 1}$ (resp. $W_{i-1} \subseteq N(R_{i - 1})$). We use the above observation frequently in the remainder of the proof. 

\sloppy We say that $\calW$ is admits an $\alpha$-bi-partition, if there is $j\in [t]$, such that $|\bigcup_{i \in [j]} W_i| \leq \alpha n$ and $|\bigcup_{i \in [t] \setminus [j-2]}W_i]| \leq \alpha n$. Note that if $\calW$ is an $\alpha$-\emph{bi-partition}, then $|N[\bigcup_{i \in [j-1]}W_i]| \leq \alpha n$ and $|N[\bigcup_{i \in [t] \setminus [j-1]}W_i]| \leq \alpha n$. If we show that $\calW$ is an $\alpha$-bi-partition, then using Lemma~\ref{lem:correct-bpc} we can conclude that $t^*\geq t$. We will use the above in later parts of our proof. 

If $|\os| \leq \beta n/2$ or $|\es| \leq \beta n/2$, then 
\SOEPC\ is better than the other. Hence, 
Lemma~\ref{lemma:alg-SOEPC} implies that $t^* \geq t$. Hereafter we assume that $|\os| > \beta n/2$ and $|\es| > \beta n/2$. The above implies that $\beta n/2 < |\os|, |\es| < (1-\beta/2)n$.  Note that there can be at most two witness sets in $\calW$ which are of size more than $\gamma n/2$, as $\gamma=0.9864$. Next, we consider cases based on the number of witness sets of size more than $\gamma n/2$ in $\cal W$.

\vspace{0.2cm}
\noindent \textbf{Case 1:} 
All witness sets in $\calW$ are of size at most $\gamma n/2$. In this case, we will show that $\calW$ admits an $\alpha$-bi-parition.
In this case \BPC\ subroutine is better than the other.
Hence, using Lemma~\ref{lem:correct-bpc}, we can conclude that $t^*\geq t$. Let $j$ be the largest integer such that $|Q_j| \leq \alpha n$. Note that $j\geq 2$, as $|W_1|+|W_2| \leq \gamma n \leq \alpha n$. The above also implies that $j < t$, as $\alpha =0.9996$. As $|Q_{j+1}| > \alpha n$, we have $|Q_j| + |W_{j+1}| > \alpha n$,  which can be rewritten as $|Q_j| > \alpha n  - |W_{j+1}|$. Note that $(Q_j, R_{j + 1})$ is a partition of $V(G)$, and thus, $|Q_j| + |R_{j+1}| = n$. Hence,
\begin{equation*}
|R_{j+1}| = n- |Q_j| < n- \alpha n  + |W_{j+1}|
\end{equation*}
We use this to obtain an upper bound on $|R_{j-1}|$. By definition, 
\begin{equation*}
|R_{j-1}| = |W_{j-1}| + |W_j| + |R_{j+1}|, \text{ and hence }\\
\end{equation*}
\begin{equation*}
|R_{j-1}|< |W_{j-1}| + |W_j| + n - \alpha n + |W_{j+1}| = n - \alpha n + |W_{j-1}| + |W_j| + |W_{j+1}|.
\end{equation*}
Since $j-1, j+1$ have the same parity (both are odd or both are even), $|W_{j-1}| + |W_{j+1}| < (1 - \beta/2)n$. From the premise of the case we have $|W_j| \leq \gamma n/2$. From the above discussions and using Ineqality~$(1)$, we get 
\begin{equation*}
|R_{j-1}|< n - \alpha n + (1 - \beta/2) n + \gamma/2 n = (2 - \alpha - \beta/2 + \gamma/2)n \le \alpha n.
\end{equation*}
 Thus, $|\bigcup_{i \in [j]}W_i| \leq \alpha n$ and $|\bigcup_{i \in [t] \setminus [j-2]}W_i]| \leq \alpha n$. Thus, $\calW$ admits an $\alpha$-bi-parition.

%The above implies that $|N[\cup_{i \in [j-1]}W_i]| \leq \alpha n$ and $|N[\cup_{i \in [t] \setminus [j-1]}W_i]| \leq \alpha n$. (Recall that $1< j < t$.) Thus, Lemma~\ref{lem:correct-bpc} implies that $t^*\geq t$. 
%As $N(R_j)$ is contained in $W_{j-1}$, cardinality of $N[R_j]$ is at most  $\alpha n$.
%This implies cardinalities of set $Q_j$ and $R_{j - 1}$ are upper bounded by $\alpha n$ and hence $\calW$ is a $\alpha$-balanced bi-partitioned witness structure.
%We prove that $\calW$ is $\alpha$-balanced bi-partitioned (Definition~\ref{def:alpha}) and hence premise of Lemma~\ref{lemma:dynamic-programming-optimal} is satisfied and main algorithm returns optimum value.

%Let $j$ be the largest integer such that the cardinality of $Q_j$ is at most $\alpha n$. By Observation~\ref{obs:singleton-end-bags}, integer $j$ is not equal to $1$ or $\ell$.
%Since $N(Q_{j-1})$ is a subset of $Q_j$, cardinality of $N[Q_{j-1}]$ is at most $\alpha n$.

\vspace{0.2cm}  
\noindent \textbf{Case 2:} There is exactly one witness set $W_k$ in $\calW$, such that $|W_k| \geq  \gamma n/2$. If $\calW$ admits an $\alpha$-bi-partition, then we can conclude that $t^*\geq t$, using Lemma~\ref{lem:correct-bpc}. Thus, we assume that $\calW$ does not admit an $\alpha$-bi-partition. Let $j$ be the largest integer, such that $|Q_j| \leq \alpha n$. (Recall that all graphs under consideration have at least $2$ vertices, $|W_1|=1$, and hence $j$ exists.) As argued previously, we have $|W_{j-1}| + |W_{j+1}| < (1 - \beta/2)n$. If $j \neq k$ then $|W_j| \leq \gamma n/2$, and arguments are similar to that of previous case. We now consider a case when $j = k$. Without loss of generality, assume that $k$ is odd. Since $|\texttt{OS}|<(1 - \beta/2) n$ and $\gamma n/2 \leq |W_j|$, we have $|\texttt{OS} \setminus W_j| \le (1 - \beta/2 - \gamma/2)n =\epsilon n$. 
In this case \SOEPC\ subroutine is better than the other. 
Thus, from Lemma~\ref{lemma:correct-NSOEPC} we can conclude that $t^*\geq t$. 

%There exists a $P_{\ell}$-witness structure, say $\calW = \{W_1, W_2, \dots, W_{\ell}\}$, which contains exactly one witness set of size greater than or equal to $\gamma n/2$.
%
%We prove that either $\calW$ is $\alpha$-balanced bi-partitioned or it is $(1 - \beta/2 - \gamma/2)$-partition concentrated (Definition~\ref{def:epsilon}) witness structure. In first case, proof of correctness is similar to that of previous case. In second case, premise of Lemma~\ref{lemma:disjoint-conn-3-optimal} is satisfied and hence the algorithm returns optimum value.
%
%Let $W_k$ be the unique witness whose cardinality is strictly greater than $\gamma n/2$. Let $j$ be the largest integer such that cardinality of $Q_j$ is at most $\alpha n$. We know that, as in previous case, $|W_{j-1}| + |W_{j+1}|$ is at most $(1 - \beta/2)n$. If $j \neq k$ then upper bound on cardinality of $W_j$ is still valid and arguments are similar as in previous case. We now consider a case when $j = k$. Without loss of generality,  assume that $k$ is an odd integer. Since cardinality of $\texttt{OS}$ is upper bounded by $(1 - \beta/2) n$ and that of $W_j$ is lower bounded by $\gamma n/2$, we have $|\texttt{OS} \setminus W_j| \le (1 - \beta/2 - \gamma/2)n$. By Observation~\ref{obs:singleton-end-bags}, we know $k$ is not equal to $1$ or $\ell$ which implies $\calW$ is a $(1 - \beta/2 - \gamma/2)$-partition concentrated set.

\vspace{0.2cm}
\noindent \textbf{Case 3:} There are exactly two witness sets $W_j,W_k$ in $\calW$, such that $|W_j|,|W_k| \geq  \gamma n/2$ and $j<k$. Consider the case when $k=j+1$. Note that in the above case, we have $|N[\bigcup_{i\in [j-1]}W_i]|, |N[\bigcup_{i\in[t]\setminus [j+1]}W_i]| \leq (1-\gamma/2) n$. Thus, from Lemma~\ref{lemma:correct-TDCPC} we can conclude that $t^*\geq t$. Now we consider the case when $j<k$ and $k\neq j+1$. We now consider the case when $j$ is odd and $k$ is even. (The case when $j$ is even and $k$ is odd can be argued analogously.) Note that $|\es\setminus W_j| \leq (1-\beta/2)n - \gamma n/2 = \epsilon n$.
In this case, \TDCPC\ subsubroutine is better than the other.
 Thus, from Lemma~\ref{lemma:correct-NSOEPC} we can conclude that $t^*\geq t$. 

Now we consider the case when $j,k$ are both even or both odd and $k\neq j+1$. Note that in the above case $k \geq j+3$. We will conclude that $t^*\geq t$, by showing that $\calW$ admits an $\alpha$-bi-partition (and Lemma~\ref{lem:correct-bpc}). To this end, we start by arguing that $|Q_{j+2}|,|R_{j+1}| \leq \alpha n$. As $k \ge j + 3$, set $W_{k}\cap Q_{j+2}=\emptyset$. Thus, $|Q_{j+2}|\leq n - \gamma n/2 \le \alpha n$. By similar arguments we can obtain that $|R_{j+1}|\le \alpha n$. Note that the above implies that $\calW$ admits an $\alpha$-bi-partitioned witness structure. This concludes the proof. 
\end{proof}

%% file: conclusion_path.tex
\section{Conclusion}

We generalized the \textsc{$2$-Disjoint Connected Subgraphs} problem, to a problem called \textsc{$3$-Disjoint Connected Subgraphs}, where instead of partitioning the vertex set into two connected sets, we are required to partition it into three connected sets. We gave an algorithm for \textsc{$3$-Disjoint Connected Subgraphs} running in time $\mcal{O}^\star(1.88^n)$. We believe that this  algorithm can be of independent interest and may find other algorithmic applications. We designed an algorithm for \textsc{Path Contraction} which breaks the $\calO^{\star}(2^n)$ barrier. It was surprising that even for a simple problem like \textsc{Path Contraction}, there was no known algorithm that solves it faster than $\calO^{\star}(2^n)$. Our algorithm for \textsc{Path Contraction} relied the fact that the number of $(Q,a,b)$-connected sets can be bounded by $\mcal{O}^\star(2^{a+b-|Q|})$. This gives us savings in the number of states that we consider in our dynamic programming routine (for enumerating partial solutions). We designed four different algorithms for \textsc{Path Contraction} and used them for appropriate instances, to obtain the main algorithm for \textsc{Path Contraction}. 

It is interesting to identify other graph contraction problems for which we can improve upon brute force algorithms. The simple algorithm described in Section~\ref{sec:intro} can be used to solve \textsc{Tree Contraction} problem. We believe there is an algorithm that breaks $\calO^{\star}(2^n)$ for \textsc{Tree Contraction}. On the other hand, we conjecture that a brute force algorithm, running in time $\calO^{\star}(n^n)$, is optimal for \textsc{Clique Contraction} under ETH. 

%% file: main.bbl
\begin{thebibliography}{10}

\bibitem{asano1983edge}
Takao Asano and Tomio Hirata.
\newblock {Edge-Contraction Problems}.
\newblock {\em Journal of Computer and System Sciences}, 26(2):197--208, 1983.

\bibitem{DBLP:journals/siamcomp/Bjorklund14}
Andreas Bj{\"{o}}rklund.
\newblock Determinant sums for undirected hamiltonicity.
\newblock {\em {SIAM} J. Comput.}, 43(1):280--299, 2014.

\bibitem{brouwer1987contractibility}
Andries~Evert Brouwer and Hendrik~Jan Veldman.
\newblock Contractibility and {NP}-completeness.
\newblock {\em Journal of Graph Theory}, 11(1):71--79, 1987.

\bibitem{DBLP:journals/talg/CyganDLMNOPSW16}
Marek Cygan, Holger Dell, Daniel Lokshtanov, D{\'{a}}niel Marx, Jesper
  Nederlof, Yoshio Okamoto, Ramamohan Paturi, Saket Saurabh, and Magnus
  Wahlstr{\"{o}}m.
\newblock On problems as hard as {CNF-SAT}.
\newblock {\em {ACM} Trans. Algorithms}, 12(3):41:1--41:24, 2016.

\bibitem{cygan2014solving}
Marek Cygan, Marcin Pilipczuk, Micha{\l} Pilipczuk, and Jakub~Onufry
  Wojtaszczyk.
\newblock Solving the 2-disjoint connected subgraphs problem faster than
  {$2^n$}.
\newblock {\em Algorithmica}, 70(2):195--207, 2014.

\bibitem{dabrowski2017contracting}
Konrad~K Dabrowski and Dani{\"e}l Paulusma.
\newblock Contracting bipartite graphs to paths and cycles.
\newblock {\em Information Processing Letters}, 127:37--42, 2017.

\bibitem{diestel-book}
Reinhard Diestel.
\newblock {\em Graph Theory, 4th Edition}, volume 173 of {\em Graduate texts in
  mathematics}.
\newblock Springer, 2012.

\bibitem{fiala2013note}
Ji{\v{r}}{\'\i} Fiala, Marcin Kami{\'n}ski, and Dani{\"e}l Paulusma.
\newblock A note on contracting claw-free graphs.
\newblock {\em Discrete Mathematics and Theoretical Computer Science},
  15(2):223--232, 2013.

\bibitem{DBLP:journals/combinatorica/FominV12}
Fedor~V. Fomin and Yngve Villanger.
\newblock Treewidth computation and extremal combinatorics.
\newblock {\em Combinatorica}, 32(3):289--308, 2012.

\bibitem{HeggernesHLP14}
Pinar Heggernes, Pim van~'t Hof, Benjamin L{\'{e}}v{\^{e}}que, and Christophe
  Paul.
\newblock Contracting chordal graphs and bipartite graphs to paths and trees.
\newblock {\em Discrete Applied Mathematics}, 164:444--449, 2014.

\bibitem{DBLP:journals/jcss/ImpagliazzoP01}
Russell Impagliazzo and Ramamohan Paturi.
\newblock On the complexity of k-sat.
\newblock {\em J. Comput. Syst. Sci.}, 62(2):367--375, 2001.

\bibitem{impagliazzo2001problems}
Russell Impagliazzo, Ramamohan Paturi, and Francis Zane.
\newblock Which problems have strongly exponential complexity?
\newblock {\em Journal of Computer and System Sciences}, 63(4):512--530, 2001.

\bibitem{kern2018contracting}
Walter Kern and Daniel Paulusma.
\newblock Contracting to a longest path in {H}-free graphs.
\newblock {\em arXiv preprint arXiv:1810.01542}, 2018.

\bibitem{DBLP:journals/talg/LokshtanovMS18}
Daniel Lokshtanov, D{\'{a}}niel Marx, and Saket Saurabh.
\newblock Known algorithms on graphs of bounded treewidth are probably optimal.
\newblock {\em {ACM} Trans. Algorithms}, 14(2):13:1--13:30, 2018.

\bibitem{DBLP:journals/siamcomp/TarjanT77}
Robert~Endre Tarjan and Anthony~E. Trojanowski.
\newblock Finding a maximum independent set.
\newblock {\em {SIAM} J. Comput.}, 6(3):537--546, 1977.

\bibitem{telle2013connecting}
Jan~Arne Telle and Yngve Villanger.
\newblock Connecting terminals and 2-disjoint connected subgraphs.
\newblock In {\em International Workshop on Graph-Theoretic Concepts in
  Computer Science}, pages 418--428. Springer, 2013.

\bibitem{van2009partitioning}
Pim van't Hof, Dani{\"e}l Paulusma, and Gerhard~J Woeginger.
\newblock Partitioning graphs into connected parts.
\newblock {\em Theoretical Computer Science}, 410(47-49):4834--4843, 2009.

\bibitem{watanabe81}
Toshimasa Watanabe, Tadashi Ae, and Akira Nakamura.
\newblock On the removal of forbidden graphs by edge-deletion or by
  edge-contraction.
\newblock {\em Discrete Applied Mathematics}, 3(2):151--153, 1981.

\bibitem{watanabe1983np}
Toshimasa Watanabe, Tadashi Ae, and Akira Nakamura.
\newblock On the {NP}-hardness of edge-deletion and contraction problems.
\newblock {\em Discrete Applied Mathematics}, 6(1):63--78, 1983.

\end{thebibliography}
